\RequirePackage{fix-cm}
\documentclass[smallcondensed,final]{svjour3}     
\smartqed  
\usepackage{mathptmx}      
\usepackage[utf8]{inputenc} 
\usepackage{amssymb}
\usepackage{amsmath, bm}
\usepackage{anyfontsize}
\usepackage{mathrsfs}
\usepackage{xcolor}
\usepackage{theorem}
\usepackage{cite}
\usepackage{url}
\usepackage{hyperref}
\usepackage{tikz}
\usepackage{pgfplots}
\usepackage{multirow}
\usepackage[nolist]{acronym}
\usetikzlibrary{plotmarks,patterns}
\usepackage{arydshln}

\usepackage[ruled,vlined,titlenumbered,linesnumbered]{algorithm2e}
\usepackage[normalem]{ulem}

\newtheorem{assumption}{Assumption}

\pgfplotsset{compat=1.17}

\newcommand{\Gal}{\mathrm{Gal}}

\newcommand{\Fqm}{\ensuremath{\mathbb F_{q^m}}}
\newcommand{\Fqms}{\ensuremath{\mathbb F_{q^{ms}}}}

\newcommand{\Fq}{\ensuremath{\mathbb F_{q}}}

\newcommand{\ZZ}{\ensuremath{\mathbb{Z}}}
\newcommand{\set}[1]{\ensuremath{\mathcal{#1}}}

\newcommand{\intervallincl}[2]{\ensuremath{[#1,#2]}}

\newcommand{\spannedBy}[1]{\ensuremath{\langle #1\rangle}}
\newcommand{\modSpan}[1]{\ensuremath{\langle #1\rangle}}

\newcommand{\intPointSet}{\ensuremath{\set{P}}}

\newcommand{\Polyring}{\ensuremath{\Fqm[x]}}
\newcommand{\Linpolyring}{\mathbb{L}_{q^m}\![x]}

\newcommand{\aut}{\ensuremath{\sigma}}
\newcommand{\der}{\ensuremath{\delta}}
\newcommand{\Id}{\ensuremath{\textsf{Id}}}
\newcommand{\SkewPolyring}{\ensuremath{\Fqm[x;\aut,\der]}}
\newcommand{\SkewPolyringZeroDer}{\ensuremath{\Fqm[x;\aut]}}
\newcommand{\remev}[2]{{#1}\!\left[#2\right]}
\newcommand{\opev}[3]{\ensuremath{{#1}(#2)_{#3}}}
\newcommand{\opfull}[2]{\ensuremath{\mathcal{D}_{#1}^{\aut,\der}(#2)}}
\newcommand{\op}[2]{\ensuremath{\mathcal{D}_{#1}(#2)}}
\newcommand{\opexp}[3]{\ensuremath{\mathcal{D}_{#1}^{#3}(#2)}}
\newcommand{\conjg}[2]{\ensuremath{{#1}^{#2}}}

\newcommand{\OCompl}[1]{\ensuremath{\mathcal{O}({#1})}}

\DeclareMathOperator{\defi}{def}
\newcommand{\defeq}{\overset{\defi}{=}}

\renewcommand{\bar}{\overline}

\newcommand{\modr}{\; \mathrm{mod}_\mathrm{r} \;}

\DeclareMathOperator{\wt}{wt}
\DeclareMathOperator{\rk}{rk}

\newcommand{\LT}[1]{\textrm{LT}(#1)}

\newcommand{\ind}[2]{\ensuremath{\textrm{Ind}_{\vec{#2}}(#1)}}
\newcommand{\mystack}[2]{\ensuremath{\genfrac{}{}{0pt}{}{#1}{#2}}}

\renewcommand{\vec}[1]{\ensuremath{\mathbf{#1}}}
\newcommand{\mat}[1]{\ensuremath{\mathbf{#1}}}

\newcommand{\lclm}[2]{\ensuremath{\text{lclm}\left(#1\right)_{#2}}}

\renewcommand{\a}{\mathbf a}
\renewcommand{\b}{\mathbf b}

\newcommand{\p}{\mathbf p}
\renewcommand{\r}{\mathbf r}

\newcommand{\w}{\mathbf w}
\newcommand{\x}{\mathbf x}

\newcommand{\B}{\mathbf B}

\newcommand{\E}{\mathbf E}

\newcommand{\Q}{\mathbf Q}
\newcommand{\R}{\mathbf R}

\newcommand{\U}{\mathbf U}

\newcommand{\0}{\mathbf 0}

\newcommand{\IntGab}[1]{\ensuremath{\mathrm{I}\mathrm{Gab}[#1]}}

\newcommand{\intSkewRS}[1]{\ensuremath{\mathrm{I}\mathrm{SRS}[#1]}}

\newcommand{\intLinRS}[1]{\ensuremath{\mathrm{I}\mathrm{LRS}[#1]}}

\newcommand{\shots}{\ensuremath{\ell}}

\newcommand{\HammingWeight}{\ensuremath{\wt_{H}}}
\newcommand{\RankWeight}{\ensuremath{\wt_{R}}}
\newcommand{\SumRankWeight}{\ensuremath{\wt_{\Sigma R}}}
\newcommand{\SkewWeight}{\ensuremath{\wt_{skew}}}

\newcommand{\RankDist}{d_R}
\newcommand{\SumRankDist}{d_{\ensuremath{\Sigma}R}}

\newcommand{\SkewDist}{d_{skew}}

\newcommand{\minpolyRem}[1]{\ensuremath{M^\text{rem}_{#1}(x)}}
\newcommand{\minpolyOp}[2]{\ensuremath{M^\text{op}_{#1}(x)_{#2}}}

\newcommand{\vecMinpolyNoX}[1]{\ensuremath{{\vec{M}}_{#1}}}
\newcommand{\vecMinpolyRem}[1]{\ensuremath{\vec{M}^\text{rem}_{#1}(x)}}
\newcommand{\vecMinpolyOp}[2]{\ensuremath{\vec{M}^\text{op}_{#1}(x)_{#2}}}

\newcommand{\bnd}[2]{\ensuremath{#1\mathopen{}\left(#2\right)\mathclose{}}}
\newcommand{\oh}[1]{\bnd{O}{#1}}
\newcommand{\softoh}[1]{\bnd{\widetilde{O}}{#1}}
\newcommand{\matmult}{\ensuremath{\omega}}
\newcommand{\OMul}[1]{\mathfrak{M}(#1)}

\newcommand{\nReceive}{\ensuremath{n_r}}

\newcommand{\degConstraint}{\ensuremath{D}}
\newcommand{\intOrder}{\ensuremath{s}}
\newcommand{\foldPar}{\ensuremath{h}}

\newcommand{\posInt}{\ensuremath{\mathbb{Z}_{+}}} 
\newcommand{\Groebner}{Gr\"obner }

\newcommand{\vecEv}[2]{\ensuremath{\mathscr{E}_{#1}(#2)}}
\newcommand{\vecEvNoInput}[1]{\ensuremath{\mathscr{E}_{#1}}}
\newcommand{\vecOpEv}[3]{\ensuremath{\mathscr{E}_{#1}^{op}(#2)_{#3}}}
\newcommand{\vecOpEvNoInput}[1]{\ensuremath{\mathscr{E}_{#1}^{op}}}
\newcommand{\vecRemEv}[2]{\ensuremath{\mathscr{E}_{#1}^{rem}(#2)}}
\newcommand{\vecRemEvNoInput}[1]{\ensuremath{\mathscr{E}_{#1}^{rem}}}

\newcommand{\intParam}{\ensuremath{s}}
\newcommand{\wOrder}{\ensuremath{\prec_{\vec{w}}}}
\newcommand{\module}[1]{\ensuremath{\set{#1}}}
\newcommand{\intModule}[1]{\ensuremath{\bar{\module{K}}_{#1}}}

\begin{acronym}
 \acro{TOP}{term-over-position}
 \acro{DaC}[D\&C]{divide-and-conquer}
 \acro{KNH}{Kötter--Nielsen--H\o holdt}
 \acro{LCLM}{least common left multiple}
 \acro{SRS}{skew Reed--Solomon}
 \acro{ISRS}{interleaved skew Reed--Solomon}
 \acro{LRS}{linearized Reed--Solomon}
 \acro{LLRS}{lifted linearized Reed--Solomon}
 \acro{ILRS}{interleaved linearized Reed--Solomon}
 \acro{LILRS}{lifted interleaved linearized Reed--Solomon}
 \acro{MSRD}{maximum sum-rank distance}
 \acro{MSD}{maximum skew distance}
 \acro{MRD}{maximum rank distance}
\end{acronym}

\DeclareMathAlphabet{\mathcal}{OMS}{cmsy}{m}{n}

\journalname{Designs, Codes and Crypography}

\begin{document}

\title{Fast Kötter--Nielsen--H\o holdt Interpolation over Skew Polynomial Rings and its Application in Coding Theory
 \thanks{Part of this work was presented at the 25th International Symposium on
Mathematical Theory of Networks and Systems (MTNS)~\cite{bartz2022fastKNH}. }
}


\titlerunning{Fast KNH Interpolation over Skew Polynomial Rings}        

\author{Hannes Bartz and Thomas Jerkovits and Johan Rosenkilde}


\institute{
	Hannes Bartz \at
	Institute of Communications and Navigation\\ 
	German Aerospace Center (DLR), D-82234 Oberpfaffenhofen, Germany\\
	Tel.: +49-8153-28-2252\\
	\email{hannes.bartz@dlr.de}
	\and 
	Thomas Jerkovits \at
	Institute of Communications and Navigation\\ 
	German Aerospace Center (DLR), D-82234 Oberpfaffenhofen, Germany\\
	Tel.: +49-8153-28-1557\\
	\email{thomas.jerkovits@dlr.de}
	\and 
	Johan Rosenkilde \at
	GitHub Inc., USA
	\\
	\email{jsrn@jsrn.dk}
}

\date{Received: date / Accepted: date}

\maketitle

\begin{abstract}
 Skew polynomials are a class of non-commutative polynomials that have several applications in computer science, coding theory and cryptography.
 In particular, skew polynomials can be used to construct and decode evaluation codes in several metrics, like e.g. the Hamming, rank, sum-rank and skew metric.

 We propose a fast divide-and-conquer variant of \acf{KNH} interpolation algorithm: it inputs a list of linear functionals on skew polynomial vectors, and outputs a reduced Gröbner basis of their kernel intersection. 
 We show, that the proposed \ac{KNH} interpolation can be used to solve the interpolation step of interpolation-based decoding of interleaved Gabidulin codes in the rank-metric, linearized Reed--Solomon codes in the sum-rank metric and skew Reed--Solomon codes in the skew metric requiring at most $\softoh{\intOrder^\matmult\OMul{n}}$ operations in $\Fqm$, where $n$ is the length of the code, $\intOrder$ the interleaving order, $\OMul{n}$ the complexity for multiplying two skew polynomials of degree at most $n$, $\matmult$ the matrix multiplication exponent and $\softoh{\cdot}$ the soft-$O$ notation which neglects log factors.
 This matches the previous best speeds for these tasks, which were obtained by top-down minimal approximant bases techniques, and complements the theory of efficient interpolation over free skew polynomial modules by the bottom-up KNH approach. 
 In contrast to the top-down approach the bottom-up KNH algorithm has no requirements on the interpolation points and thus does not require any pre-processing.
\end{abstract}

\clearpage

\section{Introduction}\label{sec:introduction}

Skew polynomials are a class of non-commutative polynomials, that were introduced by Ore in 1933~\cite{ore1933theory} and that have a variety of applications in computer science, coding theory and cryptography. 
The non-commutativity stems from the multiplication rule, which involves both, a field automorphism $\aut$ and a field derivation $\der$.
Unlike ordinary polynomials, there exist several ways to evaluate skew polynomials.
General results regarding the so-called \emph{remainder evaluation} of skew polynomials were derived in~\cite{lam1985general,lam1988vandermonde} whereas the \emph{generalized operator evaluation} was considered in~\cite{leroy1995pseudolinear}.
Depending on the choice of the automorphism $\aut$ and the derivation $\der$, skew polynomial rings (denoted by $\SkewPolyring$) include several interesting special cases, such as the ordinary polynomial ring as well as the linearized polynomial ring~\cite{Ore_OnASpecialClassOfPolynomials_1933,ore1933theory}.
This property along with the different ways to evaluate skew polynomials make them a very versatile tool with many different applications.

One important application of skew polynomials is the construction of evaluation codes, that have distance properties in several decoding metrics, including the~\emph{Hamming}, rank, sum-rank, skew and other related metrics such as the (sum-)subspace metric~\cite{boucher2014linear,martinez2018skew,martinez2019reliable,caruso2019residues}.

Many evaluation codes allow for decoding via efficient interpolation-based decoding algorithms, like e.g. the Welch--Berlekamp~\cite{welch1986error} and Sudan\cite{sudan1997decoding} algorithms for decoding Reed--Solomon codes.
In~\cite{koetter_dissertation} Kötter presented a bivariate interpolation algorithm for Sudan-like decoding of Reed--Solomon codes~\cite{koetter_dissertation} (over ordinary polynomial rings) that since then is often referred to as the \emph{Kötter interpolation}.
The Kötter interpolation as it is known today was first stated by Nielsen and H{\o}holdt~\cite{nielsen2000decoding} as a generalization of Kötter's algorithm~\cite{koetter_dissertation} which is able to handle multiplicities.
To acknowledge the contribution by Nielsen and H{\o}holdt we refer to the algorithm as \acf{KNH} interpolation.
A fast divide-and-conquer variant of the~\ac{KNH} interpolation for the Guruswami--Sudan algorithm for decoding Reed--Solomon codes was presented in~\cite{nielsen2014fast}.
Rosenkilde's algorithm~\cite{nielsen2014fast} is a bottom-up \ac{KNH}-like algorithm whose complexity is only slightly larger compared to the currently fastest approach~\cite[Section~2.5]{jeannerod2017computing}.   

A multivariate generalization of the \ac{KNH} interpolation~\cite{wang2005kotter} for free modules over ordinary polynomial rings was proposed in~\cite{wang2005kotter}.
This approach was generalized to free modules over linearized polynomial rings in~\cite{xie2011general}. 
A generalization of the multivariate \ac{KNH} interpolation to free modules over skew polynomial rings was proposed in~\cite{liu2014kotter}, which contains the variants over ordinary polynomial rings~\cite{wang2005kotter} and linearized polynomial rings~\cite{xie2011general} as a special case.

The evaluation and interpolation of multivariate skew polynomials was also considered in~\cite{martinez2019evaluation}; here with the main motivation to construct Reed-Muller-like codes (see also~\cite{geiselmann2019skew,augot2021rank,martinez2022theory}).

\subsection{Main contribution}

In this paper, we propose a fast~\ac{DaC} variant of the \ac{KNH} interpolation in skew polynomial rings~\cite{liu2014kotter}, that uses ideas from~\cite{nielsen2014fast}.
The main idea of the proposed algorithm (Algorithm~\ref{alg:skewIntTree}) is, that the interpolation problem is divided into smaller sub-problems, that can be solved and merged efficiently.
In particular, the update operations in each loop of the \ac{KNH} interpolation are ``recorded'' and then applied to a degree-reduced basis in the merge step rather than to a non-reduced basis.
This allows to control the degree of the polynomials during the interpolation procedure which in turn results in a lower computational complexity.

We state the interpolation problem and the algorithm in a general way using linear functionals over skew polynomials rings with arbitrary automorphisms and derivations.
We show how the fast \ac{KNH} interpolation can be applied to interpolation-based decoding of (interleaved) Gabidulin codes~\cite{Loidreau_Overbeck_Interleaved_2006,overbeck2006decoding}, interleaved linearized Reed--Solomon codes~\cite{martinez2018skew,caruso2019residues,bartz2021decoding,bartz2022fast} and (interleaved) skew Reed--Solomon codes~\cite{boucher2014linear,bartz2021decoding,bartz2022fast}.

We consider skew polynomials over finite fields only. 
However, the results (except for the complexity statements) also hold for skew polynomials over arbitrary finite Galois extensions $\mathbb{L}/\mathbb{K}$ instead of $\Fqm/\Fq$ and automorphisms $\aut \in \Gal(\mathbb{L}/\mathbb{K})$ with $\mathbb{K} = \mathbb{L}^\aut$ and derivations $\der:\mathbb{K}\mapsto\mathbb{K}$ satisfying~\eqref{eq:def_derivation} for all $a,b\in\mathbb{K}$.

For the above mentioned applications using generalized operator and remainder evaluation maps over skew polynomial rings with arbitrary field automorphisms we discuss the asymptotic complexity for zero derivations ($\der=0$). 
The asymptotic complexity for solving the interpolation step with the proposed approach is $\softoh{\intOrder^\matmult\OMul{n}}$ operations in $\Fqm$, where $n$ is the length of the code, $\intOrder$ a decoding parameter (e.g. interpolation order, usually $\intOrder\ll n$), $\OMul{n}\in\softoh{n^{1.635}}$ the complexity for multiplying two skew polynomials of degree at most $n$, $\matmult$ the matrix multiplication exponent (currently $\omega < 2.37286$) and $\softoh{\cdot}$ denotes the soft-$O$ notation which neglects log factors.

The original skew~\ac{KNH} interpolation from~\cite{liu2014kotter} has an asymptotic complexity of $\oh{\intOrder^2n^2}$ operations in $\Fqm$, which is larger compared to the proposed approach for most practical cases where we usually have $\intOrder \ll n$.

The interpolation step of the above mentioned coding applications can also be solved using the skew minimal approximant bases methods from~\cite{bartz2021fast} requiring at most $\softoh{\intOrder^\matmult\OMul{n}}$ operations in $\Fqm$.
This approach can be seen as a top-down approach: first construct a module that contains all solutions to the interpolation problem, and then find the minimal solution in that module. 
The \ac{KNH} family of algorithms are bottom-up: Gradually build up a minimal basis solving the interpolation constraints one by one. 
Developing both approaches in tandem has been very fruitful for the analogous family of algorithms for ordinary polynomial rings.

Due to the top-down nature, the minimal approximant bases method~\cite[Algorithm~6]{bartz2021fast} requires an additional step if the first entries of the interpolation points (related to the generalized operator evaluation maps) are not linearly independent (see~\cite[Theorem~22]{bartz2021fast}), which is \emph{not} required in the proposed KNH-like algorithm.

A comparison between the proposed fast \ac{KNH} interpolation algorithm and existing interpolation methods for ordinary and skew polynomial rings, including the computational complexity for the zero-derivation case ($\der=0$), is given in Table~\ref{tab:complexity_overview}.
Note, that for ordinary polynomial rings we have that $\OMul{n}\in\softoh{n}$. Additionally, in the ordinary polynomial ring case we have the notion of multiplicities of roots; a concept which is not yet generalized to the skew polynomial setting and therefore omitted.

Although the proposed interpolation algorithm achieves the best known complexity over skew polynomial rings, Table~\ref{tab:complexity_overview} shows that the best computational complexity of $\softoh{\intOrder^{\matmult-1}\OMul{n}}$ for ordinary polynomial rings~\cite{jeannerod2017computing} is not yet reached for the skew polynomial case.
Closing this gap is a potential topic for future work.

\begin{table}[ht!]
\caption{Overview of the computational complexity of the proposed fast \ac{KNH} interpolation approach compared to existing methods for the case of zero derivations ($\der=0$). 
Here $n$ is the number of interpolation points, $\intOrder$ an interpolation parameter (usually $\intOrder\ll n$), $\OMul{n}\in\softoh{n^{1.635}}$ the complexity for multiplying two skew polynomials of degree at most $n$ and $\matmult$ the matrix multiplication exponent (currently $\omega < 2.37286$). For ordinary polynomial rings we have that $\OMul{n}\in\softoh{n}$.}
\label{tab:complexity_overview}
\centering
\renewcommand{\arraystretch}{1.5}
\begin{tabular}{|c|l|c|c|c|}
 \hline
 \multicolumn{2}{|c|}{Interpolation Method} & Type & Polynomial Ring & Complexity ($\der=0$)
 \\ \hline\hline
 \multirow{4}{*}{\rotatebox[origin=c]{90}{ordinary}} & \ac{KNH}~\cite{wang2005kotter} & bottom-up & $\Fqm[x;\Id,0]$ & $\oh{\intParam^2n^2}$
 \\ \cline{2-5}
 & \ac{DaC} \ac{KNH}~\cite{nielsen2014fast} & bottom-up & $\Fqm[x;\Id,0]$ & $\softoh{\intOrder^\matmult\OMul{n}}$
 \\ \cline{2-5}
 & Min. approximant bases method~\cite{giorgi2003complexity} & top-down & $\Fqm[x;\Id,0]$ & $\softoh{\intOrder^\matmult\OMul{n}}$
 \\ \cline{2-5}
 & Min. interpolation bases~\cite{jeannerod2017computing} & top-down & $\Fqm[x;\Id,0]$ & $\softoh{\intOrder^{\matmult-1}\OMul{n}}$
 \\ \hline \hline
 \multirow{4}{*}{\rotatebox[origin=c]{90}{skew}} & Linearized \ac{KNH}~\cite{xie2011general} & bottom-up & $\Fqm[x;\aut_\text{Frob},0]$ & $\oh{\intParam^2n^2}$
 \\ \cline{2-5}
 & Skew \ac{KNH}~\cite{liu2014kotter} & bottom-up & $\SkewPolyring$ & $\oh{\intParam^2n^2}$
 \\ \cline{2-5}
 & \textbf{\ac{DaC} skew \ac{KNH} (this contribution)} & bottom-up & $\SkewPolyring$ & $\softoh{\intOrder^\matmult\OMul{n}}$ 
 \\ \cline{2-5}
 & Skew min. approximant bases method~\cite{bartz2021fast} & top-down & $\SkewPolyring$ & $\softoh{\intOrder^\matmult\OMul{n}}$
 \\ \hline 
\end{tabular}
\end{table}

\subsection{Outline of the paper}

Section~\ref{sec:preliminaries} gives definitions and notations related to skew polynomials as well as a definition of the skew~\ac{KNH} interpolation algorithm from~\cite{liu2014kotter}.
Section~\ref{sec:fast_skew_knh} presents a fast general~\ac{DaC} framework for the skew~\ac{KNH} interpolation.
Section~\ref{sec:coding_applications} considers the application of the fast skew~\ac{KNH} interpolation for decoding (interleaved) Gabidulin, linearized Reed--Solomon and Skew Reed--Solomon codes. The complexity analysis shows that we obtain the currently fastest known decoders for the considered codes.
Section~\ref{sec:conclusion} concludes the paper.

\section{Preliminaries}\label{sec:preliminaries}

\subsection{Sets, Vectors and Matrices over Finite Fields}
Let $\Fq$ be a finite field and denote by $\Fqm$ the extension field of degree $m$. 
Sets are denoted by $\set{A}=\{a_0,a_1,\dots,a_{n-1}\}$. 
The cardinality of a set $\set{A}$ is denoted by $|\set{A}|$.

Vectors and matrices over $\Fqm$ are denoted by bold lower-case and upper-case letters such as $\vec{a}$ and $\mat{A}$, respectively, and the elements are indexed beginning from zero.
The $r\times c$ zero matrix is denoted by $\mat{0}_{r\times c}$ and the $r\times r$ identity matrix is denoted by $\mat{I}_r$.
The $i$-th row of a matrix $\mat{A}$ is denoted by $\vec{a}_i$.
Let $\Fqm^N$ denote the set of all row vectors of length $N$ over $\Fqm$ and let $\Fqm^{M\times N}$ denote the set of all $M\times N$ matrices over $\Fqm$.
The rank of a matrix $\mat{A}\in\Fq^{M\times N}$ is denoted by $\rk_q(\mat{A})$.
Under a fixed basis of $\Fqm$ over $\Fq$, there is a bijection between a vector $\vec{a}\in\Fqm^n$ and a matrix $\Fq^{m\times n}$.
This allows us to define the rank of a vector as $\rk_q(\vec{a})\defeq\rk(\mat{A})$ where $\mat{A}$ is the corresponding matrix of $\vec{a}$ over $\Fq$.
The \emph{Hamming} weight of a vector $\vec{a}\in\Fqm^N$ is defined as
\begin{equation}
 \HammingWeight(\vec{a})=|\{a_i\in\vec{a}:a_i\neq 0\}|.
\end{equation}

Let $\aut:\Fqm\mapsto\Fqm$ be a field automorphism of $\Fqm$ and let $\der:\Fqm\mapsto\Fqm$ be a $\aut$-derivation such that
\begin{align}\label{eq:def_derivation} 
	\der(a+b)=\der(a)+\der(b)
   \quad\text{and}\quad
	\der(ab)=\der(a)b+\aut(a)\der(b).
\end{align}
Over a finite field, all $\aut$-derivations are of the form (see e.g.~\cite[Proposition~1]{liu2014kotter})
\begin{equation}
	\der(a)=b(\aut(a)-a) \qquad \text{for } b\in\Fqm. 
\end{equation}

For any two elements $a\in\Fqm$ and $c\in\Fqm^*=\Fqm\setminus\{0\}$ define
\begin{equation}
 a^c\defeq \aut(c)ac^{-1}+\der(c)c^{-1}
\end{equation}
where term $\der(c)c^{-1}$ is called the \emph{logarithmic derivative} of $c$.
Two elements $a,b\in\Fqm$ are called $(\aut,\der)$-conjugates, if there exists an element $c\in\Fqm^*$ such that $b=a^c$. Otherwise, $a$ and $b$ are called $(\aut,\der)$-distinct.
The notion of $(\aut,\der)$-conjugacy defines an equivalence relation on $\Fqm$ and thus a partition of $\Fqm$ into conjugacy classes (see~\cite{lam1988vandermonde}).

\begin{definition}[Conjugacy Class~\cite{lam1988vandermonde}]
 The set 
 \begin{equation}
 	\set{C}(a)\defeq\left\{a^c:c\in\Fqm^*\right\}
 \end{equation}
 is called \emph{conjugacy class} of $a$.
\end{definition}

\subsection{Skew Polynomials}

\emph{Skew polynomials} are non-commutative polynomials that were introduced by Ore~\cite{ore1933theory}. 
The set of all polynomials of the form 
\begin{equation}
	f(x)=\sum_{i}f_ix^i \quad \text{with} \quad f_i\in\Fqm
\end{equation}
together with the ordinary polynomial addition and the multiplication rule
\begin{equation}\label{eq:mult_rule}
	xa=\aut(a)x+\der(a)
\end{equation}
forms the non-commutative ring of \emph{skew polynomials} that is denoted by $\SkewPolyring$.
The \emph{degree} of a skew polynomial $f\in\SkewPolyring$ is defined as $\deg(f)\defeq\max_i\{i:f_i\neq0\}$ for $f\neq0$ and $-\infty$ else.
Further, by $\SkewPolyring_{<n}$ we denote the set of skew polynomials from $\SkewPolyring$ of degree less than $n$.

The \emph{monic}~\ac{LCLM} of some polynomials $p_0,p_1,\dots,p_{n-1}\in\SkewPolyring$ is denoted by
\begin{equation}\label{eq:def_lclm}
 \lclm{p_i}{0\leq i\leq n-1}\defeq\lclm{p_0,p_1,\dots,p_{n-1}}{}.
\end{equation} 

The skew polynomial ring $\SkewPolyring$ is a left and right Euclidean domain, i.e., for any $f,g\in\SkewPolyring$ with $\deg(f)\geq\deg(g)$ there exist unique polynomials $q_L,r_L,q_R,r_R\in\SkewPolyring$ such that
\begin{equation}
	f(x)=q_R(x)g(x)+r_R(x)=g(x)q_L(x)+r_L(x)
\end{equation}
where $\deg(r_R),\deg(r_L)<\deg(g)$ (see~\cite{ore1933theory}).
Efficient Euclidean-like algorithms for performing left/right skew polynomial division exist~\cite{caruso2017fast,caruso2017new,puchinger2017fast}.
For two skew polynomials $f,g\in\SkewPolyring$, denote by $f\modr g$ the remainder of the right division of $f$ by $g$.

\begin{example}
 Applying the multiplication rule in~\eqref{eq:mult_rule} to $x^2a$ we get
 \begin{align*}
  x^2a &= x(\aut(a)x+\der(a)) = x\aut(a)x+x\der(a) 
  \\
  &= (\aut^2(a)x+\der(\aut(a)))x+\aut(\der(a))x+\der^2(a)
  \\
  &=\aut^2x^2+(\der(\aut(a))+\aut(\der(a))x+\der^2(a).
 \end{align*}
\end{example}

There are several interesting cases where skew polynomial rings coincide with other polynomial rings:
\begin{itemize}
 \item For $\aut$ being the identity and $\der$ being zero derivation (i.e. for $\der(a)=0$ for all $a\in\Fqm$) we have that $\SkewPolyring$ is equivalent to the ordinary polynomial ring $\Polyring$.

 \item For $\der$ being the zero derivation we get the \emph{twisted} polynomial ring $\SkewPolyringZeroDer$.

 \item For $\aut$ being the identity we get the differential polynomial ring $\Fqm[x,\der]$.

 \item For $\aut$ being the Frobenius automorphism of $\Fqm$ (i.e. $\aut(\cdot)=\cdot^q$) and $\der$ being the zero derivation we have that $\SkewPolyringZeroDer$ is isomorphic to the linearized polynomial ring $\Linpolyring$~\cite{Ore_OnASpecialClassOfPolynomials_1933,ore1933theory}.
\end{itemize}

There exist two variants of skew polynomial evaluation: the \emph{(generalized) operator} evaluation and the \emph{remainder} evaluation.

\subsubsection{Generalized Operator Evaluation}

The generalized operator evaluation defined in~\cite{leroy1995pseudolinear} allows to $\Fq$-linearize the skew polynomial evaluation and therefore establishes the link between the skew polynomial ring and the linearized polynomial ring~\cite{Ore_OnASpecialClassOfPolynomials_1933,ore1933theory}.

Given an $\Fqm$ automorphism $\aut$, a $\aut$-derivation $\der$ and an element $a\in\Fqm$, the $(\aut,\der)$ operator $\opfull{a}{b}:\Fqm\mapsto\Fqm$ is defined as
\begin{equation}\label{eq:def_op}
 \opfull{a}{b}\defeq\aut(b)a+\der(b),\qquad\forall b\in\Fqm.	
\end{equation} 
We use the notation $\op{a}{b}$ whenever $\aut$ and $\der$ are clear from the context.
For an integer $i\geq0$, we define $\opexp{a}{b}{i+1}=\op{a}{\opexp{a}{b}{i}}$ and $\opexp{a}{b}{0}=b$.
\begin{definition}[Generalized Operator Evaluation~\cite{martinez2018skew}]
For a skew polynomial $f\in\SkewPolyring$ the generalized operator evaluation $\opev{f}{b}{a}$ of $f$ at an element $b\in\Fqm$ w.r.t. the evaluation parameter $a\in\Fqm$ is defined as
\begin{equation}\label{eq:opEval}
	\opev{f}{b}{a}\defeq\sum_{i}f_i\opexp{a}{b}{i}.
\end{equation}
\end{definition}
The definition of the generalized operator evaluation includes the operator evaluation $\opev{f}{b}{1}$ as a special case ($a=1$).

The generalized operator evaluation is an $\Fq$-linear map, i.e. for any $f\in\SkewPolyring$, $\lambda_1,\lambda_2\in\Fq$ and $a, b_1,b_2\in\Fqm$ we have that (see~\cite[Lemma~23]{martinez2018skew} and~\cite{lam1994hilbert})
\begin{equation}
 \opev{f}{\lambda_1b_1+\lambda_2b_2}{a}=\lambda_1\opev{f}{b_1}{a}+\lambda_2\opev{f}{b_2}{a}.
\end{equation}
For a vector $\vec{b}=(b_0,b_1,\dots,b_{n-1})\in\Fqm^n$ we define the generalized multipoint operator evaluation of a skew polynomial $f\in\SkewPolyring$ w.r.t. an $a\in\Fqm$ as
\begin{equation}
	\opev{f}{\vec{b}}{a}\defeq\left(\opev{f}{b_0}{a},\opev{f}{b_1}{a},\dots,\opev{f}{b_{n-1}}{a}\right).
\end{equation}

For a set $\set{B}=\{b_0,b_1,\dots,b_{n-1}\}\subseteq\Fqm$ and a vector $\vec{a}=(a_0,a_1,\dots,a_{n-1})\in\Fqm^n$ the minimal skew polynomial that vanishes on all elements in $\set{B}$ w.r.t. the evaluation parameters in $\vec{a}$ is defined as (see e.g.~\cite{caruso2019residues})
\begin{equation}\label{eq:def_minpoly_gen_op}
 \minpolyOp{\set{B}}{\vec{a}}
 =\lclm{x-\frac{\aut(b_i)a_i}{b_i}}{\mystack{0\leq i\leq n-1}{b_i\neq 0}}.
\end{equation}
The degree of $\minpolyOp{\set{B}}{\vec{a}}$ satisfies
\begin{equation}
	\deg(\minpolyOp{\set{B}}{\vec{a}})\leq n
\end{equation}
where equality holds if the sequence of elements $b_i$ that have the same evaluation parameter $a_i$ are $\Fq$-linearly independent and the distinct evaluation parameters $a_i$ are from different conjugacy classes (see~\cite{caruso2019residues}).

\begin{example}
 Consider the set $\set{B}=\{b_0,b_1,b_2,b_3\}\subseteq\Fqm$ and a vector $\a=(a_0,a_1,a_2,a_3)\in\Fqm^4$, where $a_0 = a_1$ and $a_2 = a_3$ are representatives from different conjugacy classes.
 Then we have $\deg(\minpolyOp{\set{B}}{\vec{a}})= n$ if and only if $b_0$ and $b_1$ are $\Fq$-linearly independent and $b_2$ and $b_3$ are $\Fq$-linearly independent.
\end{example}

Similar to ordinary polynomials, we get the following result for the generalized operator evaluation of a polynomial modulo a particular minimal polynomial.

\begin{lemma}\label{lem:mod_op_eval}
 For any $f\in\SkewPolyring$, $\set{B}=\{b_0,b_1,\dots,b_{n-1}\}\subseteq\Fqm$ and $\a=(a_0,a_1,\dots,a_{n-1})\in\Fqm^n$ we have that
 \begin{equation}\label{eq:mod_gen_op_eval}
 	\opev{f}{b_i}{a_i}=\opev{\left(f(x)\modr\minpolyOp{\set{B}}{\a}\right)}{b_i}{a_i},
    \quad\forall i=0,\dots,n-1.
 \end{equation}
\end{lemma}

\begin{proof}
 Since $\SkewPolyring$ is a left/right Euclidean domain, there exist two unqiue polynomials $q,r\in\SkewPolyring$ such that
 \begin{equation}
 	f(x)=q(x)\minpolyOp{\set{B}}{\a}+r(x) 
 	\quad\Longleftrightarrow\quad
 	r(x)=f(x)\modr\minpolyOp{\set{B}}{\a}.
 \end{equation}
 Since $\minpolyOp{\set{B}}{\a}$ vanishes on all $b_i$ w.r.t. $a_i$ for all $i=0,\dots,n-1$ we have that
 \begin{equation}
  \opev{f}{b_i}{a_i}=\opev{r}{b_i}{a_i},
  \quad\forall i=0,\dots,n-1,
 \end{equation}
 and the result follows.
 \qed
\end{proof}

\subsubsection{Remainder Evaluation}

Another variant of skew polynomial evaluation is the \emph{remainder evaluation} defined in~\cite{lam1985general,lam1988vandermonde}, which generalizes the concept of polynomial evaluation by means of (right) division.

 \begin{definition}[Remainder Evaluation~\cite{lam1985general,lam1988vandermonde}]
  For a skew polynomial $f\in\SkewPolyring$ the remainder evaluation $\remev{f}{b}$ of $f$ at an element $b\in\Fqm$ is defined as the unique remainder of the right division of $f(x)$ by $(x-b)$ such that
  \begin{equation}\label{eq:remEval}
  	f(x)=g(x)(x-b)+\remev{f}{b} \quad \Longleftrightarrow\quad \remev{f}{b} = f(x) \modr (x-b).
  \end{equation}
 \end{definition}

For a vector $\vec{b}=(b_0,b_1,\dots,b_{n-1})\in\Fqm^n$ we define the  multipoint remainder evaluation of a skew polynomial $f\in\SkewPolyring$ as
\begin{equation}
	\remev{f}{\vec{b}}\defeq\left(\remev{f}{b_0},\remev{f}{b_1},\dots,\remev{f}{b_{n-1}}\right).
\end{equation}

In the following we recall important properties of the concept of $P$-independence (or \emph{polynomial} independence) from~\cite{lam1985general, lam1988algebraic,martinez2018skew}.
Given a set $\set{F}\subseteq\SkewPolyring$ we define its \emph{zero set} $\set{Z}(\set{F})\subseteq\Fqm$ as
\begin{equation}\label{eq:def_zero_set}
 \set{Z}(\set{F})\defeq\{a\in\Fqm:\remev{f}{a}=0,\forall f\in\set{F}\}.
\end{equation}
For a set $\set{B}=\{b_0,b_1,\dots,b_{n-1}\}\subseteq\Fqm$ we define its \emph{associated ideal} as (see~\cite{martinez2018skew})
\begin{equation}\label{eq:def_ass_ideal}
 I(\set{B})=\{f\in\SkewPolyring:\remev{f}{b}=0,\forall b\in\set{B}\}.
\end{equation}

For a set $\set{B}=\{b_0,b_1,\dots,b_{n-1}\}\subseteq\Fqm$ the unique minimal skew polynomial that vanishes on all elements in $\set{B}$ w.r.t. the remainder evaluation is defined as (see e.g.~\cite{boucher2014linear})
\begin{equation}\label{eq:def_minpoly_rem}  
 \minpolyRem{\set{B}}=\lclm{x-b_i}{0\leq i\leq n-1}.
\end{equation}
Since $\minpolyRem{\set{B}}$ has the minimal degree among all polynomials in $\SkewPolyring$ that vanish on $\set{B}$ it generates the \emph{left} $\SkewPolyring$-ideal $I(\set{B})$.
The degree of $\minpolyRem{\set{B}}$ satisfies
\begin{equation}
 \deg(\minpolyRem{\set{B}})\leq n
\end{equation}
where the elements $b_0,b_1,\dots,b_{n-1}$ are called $P$-independent (or \emph{polynomially independent}) if and only if $\deg(\minpolyRem{\set{B}})=n$. 

The \emph{closure} of a set $\set{B}\subseteq\Fqm$ is defined as the zero set of its minimal polynomial, i.e. as
\begin{equation}\label{eq:def_closure}
 \bar{\set{B}}\defeq\set{Z}(I(\set{B}))=\set{Z}(\minpolyRem{\set{B}})
\end{equation}
where $\set{B}$ is called $P$-closed if and only if $\bar{\set{B}}=\set{B}$.
For a $P$-closed set $\set{B}$ it can be shown that any root of $\minpolyRem{\set{B}}$ is an $(\aut,\der)$-conjugate of an element in $\set{B}$ (see~\cite{lam1988algebraic}).

Similar to the result w.r.t. to the generalized operator evaluation in Lemma~\ref{lem:mod_op_eval}, we obtain the following result w.r.t. the remainder evaluation.

\begin{lemma}\label{lem:mod_rem_eval}
 For any $p\in\SkewPolyring$ and $\set{B}=\{b_0,b_1,\dots,b_{n-1}\}\subseteq\Fqm$ we have that
 \begin{equation}\label{eq:mod_rem_eval}
 	\remev{p}{b_i}=\remev{\left(p(x)\modr\minpolyRem{\set{B}}\right)}{b_i},
    \quad\forall i=0,\dots,n-1.
 \end{equation}
\end{lemma}

\begin{proof}
 Since $\SkewPolyring$ is a left/right Euclidean domain, there exist two unique polynomials $q,r\in\SkewPolyring$ such that
 \begin{equation}
 	p(x)=q(x)\minpolyRem{\set{B}}+r(x) 
 	\quad\Longleftrightarrow\quad
 	r(x)=p(x)\modr\minpolyRem{\set{B}}.
 \end{equation}
 Since $\minpolyRem{\set{B}}$ vanishes on all $b_i\in\set{B}$ we have that $\remev{p}{b_i}=\remev{r}{b_i}$ for all $i=0,\dots,n-1$ and the result follows.
 \qed
\end{proof}

The following result from~\cite{leroy1995pseudolinear} (see also~\cite{martinez2018skew}) shows the relation between the generalized operator and the remainder evaluation.

\begin{lemma}[Connection between Evaluation Types~\cite{martinez2018skew,leroy1995pseudolinear}]\label{lem:rel_remev_opev}
 For any $a\in\Fqm$, $b\in\Fqm^*$ and $f\in\SkewPolyringZeroDer$ we have that
 \begin{equation}
 	\remev{f}{\op{a}{b}b^{-1}}b=\opev{f}{b}{a}.
 \end{equation}
\end{lemma}

\subsubsection{Skew Polynomial Vectors and Matrices}

To be consistent with the conventional notation in coding theory we denote both, vectors and matrices, by bold letters. The dimensions are clear from the context.

For two vectors $\vec{a},\vec{b}\in\SkewPolyring^n$ we denote the element-wise right modulo operation by
\begin{equation}
    \vec{a}\modr\vec{b}\defeq\left(a_0\modr b_0,a_1\modr b_1,\dots,a_{n-1}\modr b_{n-1}\right).
\end{equation}

For two vectors $\vec{a},\vec{b}\in\SkewPolyring^n$ we define the element-wise \ac{LCLM} as 
\begin{equation}
 \lclm{\vec{a},\vec{b}}{}\defeq\left(\lclm{a_0,b_0}{},\lclm{a_1,b_1}{},\dots,\lclm{a_{n-1},b_{n-1}}{}\right).
\end{equation}

For a vector $\vec{a}=\left(a_0,a_1,\dots,a_{n-1}\right)\in\SkewPolyring^n$ and a vector $\vec{w}=(w_0,w_1,\dots,w_{n-1})\in\mathbb{Z}_+^n$ we define its $\vec{w}$-weighted degree as
\begin{equation}\label{eq:def_w_deg} 
 \deg_\vec{w}(\vec{a})\defeq\max_{0 \leq j \leq n-1}\{\deg(a_j)+w_j\}.
\end{equation}
Further, we define the $\vec{w}$-weighted monomial ordering $\wOrder$ on $\SkewPolyring^n$ such that we have
\begin{equation}\label{eq:top_order}
	x^\ell\vec{e}_j\wOrder x^{\ell'}\vec{e}_{j'}
\end{equation}
if $\ell+w_j<\ell'+w_{j'}$ or if $\ell+w_j=\ell'+w_{j'}$ and $j<j'$, where $\vec{e}_j$ denotes the $j$-th unit vector over $\SkewPolyring$.
The definition of $\wOrder$ coincides with the $\vec{w}$-weighted \ac{TOP} ordering as defined in~\cite{adams1994introduction}.

For a vector $\vec{a} \in \SkewPolyring^{n} \setminus \{\0\}$ and weighting vector $\vec{w} = (w_0,\dots,w_{n-1}) \in \mathbb{Z}_+^n$, we define the $\vec{w}$-pivot index $\ind{\a}{\w}$ of $\vec{a}$ to be the largest index $j$ with $0 \leq j \leq n-1$ such that $\deg(a_j) + w_j = \deg_{\vec{w}}(\vec{a})$.

For a nonzero vector $\in \SkewPolyring^{n}$ we identify the \emph{leading term} $\LT{\a}$ of $\a$ as the maximum term $a_{i,j}x^j$ under $\wOrder$.
Note, that in this case $j$ coincides with the $\vec{w}$-pivot index of $\vec{a}$.

A matrix $\mat{A}\in\SkewPolyring^{a\times b}$ with $a \leq b$ is in (row) $\vec{w}$-ordered weak Popov form if the $\vec{w}$-pivot indices of its rows are strictly increasing in the row index~\cite{mulders2003lattice}.

A free $\SkewPolyring$-module is a module that has a basis that consists of $\SkewPolyring$-linearly independent elements. The rank of this module equals the cardinality of that basis.

In the following we consider particular bases for (left) $\SkewPolyring$-modules.

\begin{definition}[$\vec{w}$-ordered weak-Popov Basis~\cite{bartz2021fast}]
	Consider a \emph{left} $\SkewPolyring$-submodule $\module{M}$ of $\SkewPolyring^b$.
	For $\vec{w} \in \ZZ^a$, a left $\vec{w}$-ordered weak-Popov basis is a full-rank matrix $\mat{A}\in\SkewPolyring^{a \times a}$ s.t.
	\begin{enumerate}
		\item $\mat{A}$ is in $\vec{w}$-ordered weak Popov form.
		\item The rows of $\mat{A}$ are a basis of $\module{M}$.
	\end{enumerate}
\end{definition}

We will now establish a connection between $\vec{w}$-ordered weak-Popov Bases and Gröbner bases w.r.t. $\wOrder$ for left $\SkewPolyring$-submodules.
This connection is well-known for ordinary commutative polynomial rings (see e.g.~\cite{fitzpatrick1995key,alekhnovich2002linear,kojima2007canonical,nielsen2013list,neiger2016bases}).
For skew polynomial rings this relation was derived in~\cite[Chapter~6]{middeke2012computational} and also used in~\cite{bartz2021fast}.

For a short introduction to Gröbner bases the reader is referred to~\cite{sturmfels2005groebner}.
An extensive study of Gröbner bases can be found in~\cite{cox1992ideals}.

\begin{definition}[\Groebner Basis~\cite{cox1992ideals}]
 Let $\module{M}$ be a left $\SkewPolyring$-submodule. 
 A subset $\set{B}=\{\b_0,\b_1,\dots,\b_{\nu-1}\}\subset \module{M}$ is called a Gröbner basis for $\module{M}$ under $\wOrder$ if the leading terms of $\module{B}$ span a left module that contains all leading terms in $\module{M}$, i.e. if $\modSpan{\LT{b_0},\dots,\LT{b_{\nu-1}}}=\modSpan{\LT{\module{M}}}$.
\end{definition}

A Gröbner basis for a $\SkewPolyring$-submodule $\module{M}$ is not necessarily a \emph{minimal} generating set for $\module{M}$ since any subset of $\module{M}$ that contains a Gröbner basis is also a Gröbner basis (see~\cite{cox1992ideals,sturmfels2005groebner}). 
The following definition imposes a minimality requirement on the cardinality of \Groebner bases for an $\SkewPolyring$-submodule under $\wOrder$.

\begin{definition}[Minimal Gröbner Basis~\cite{cox1992ideals}]\label{def:minGroebner}
 Given a monomial ordering $\wOrder$, a Gröbner basis $\set{B}$ for a left $\SkewPolyring$-sub\-module $\module{M}$ is called minimal if for all $\p\in\set{B}$ the leading term $\LT{\p}$ is not contained in the module $\modSpan{\LT{\set{B}\setminus\{\p\}}}$, i.e. if $\LT{\p}\notin\modSpan{\LT{\set{B}\setminus\{\p\}}}$. 
\end{definition}

A minimal Gröbner basis $\set{B}$ w.r.t. to $\wOrder$ is called \emph{reduced} Gröbner basis if all leading terms are normalized and no monomial of $\p\in\set{B}$ is in $\modSpan{\LT{\set{B}\setminus\{\p\}}}$.	

Although~\cite[Theorem~6.29]{middeke2012computational} establishes the connection between the stronger $\w$-ordered Popov form and the corresponding \emph{reduced} Gröbner basis w.r.t. $\wOrder$, the arguments also hold for the relation between the $\w$-ordered weak-Popov form and the minimal Gröbner basis w.r.t. $\wOrder$.
	
Note that given a module monomial order $\prec$ and a basis $\set{B} \subset \SkewPolyring^n$ of a submodule $\module{M}$ there exist an efficient method to determine a weighting vector $\vec{w}$ and a column permutation $P$ such that the weak Popov form under $\wOrder$ of $P(\mathcal{M})$ equals the $P$-permuted \emph{minimal} Gröbner basis of $\module{M}$ under $\prec$ (see~\cite[Chapter~1.3.4]{neiger2016bases}).

\subsubsection{Cost Model for Skew Polynomial Operations}

For deriving the computational complexity we consider only skew polynomials with zero derivations, i.e. only skew polynomials from $\SkewPolyringZeroDer$.

We use the big-O notation $\oh{\cdot}$ to state asymptotic costs of algorithms. Further, we use shorthand $\softoh{\cdot}$ for $f(n) \in \softoh{g(n)} \Leftrightarrow \exists k: f(n) \in \oh{g(n)\log^k g(n)}$ which is equivalent to the $\oh{\cdot}$ notation, ignoring logarithmic factors in the input parameter.
We denote by $\matmult$ the matrix multiplication exponent, i.e.~the infimum of values $\matmult_0 \in [2; 3]$ such that there is an algorithm for multiplying $n \times n$ matrices over $\Fqm$ in $O(n^{\matmult_0})$ operations in $\Fqm$. The currently best known cost bound in operations in $\Fqm$ is $\matmult < 2.37286$ (see~\cite{puchinger2017fast}).

By $\OMul{n}$ we denote the cost of multiplying two skew polynomials from $\SkewPolyringZeroDer$ of degree $n$. The currently best known cost bound in operations for $\OMul{n}$ in $\Fqm$ is (see~\cite{puchinger2017fast})
\begin{equation*}
\OMul{n} \in \oh{n^{\min\left\{\frac{\matmult+1}{2},1.635\right\}}}.
\end{equation*}
There exist other fast algorithms whose complexity is stated w.r.t. operations in $\Fq$ (see~\cite{caruso2017fast,caruso2017new}).
Hence, the following skew polynomial operations in $\SkewPolyringZeroDer$ can be performed in $\softoh{\OMul{n}}$:
 \begin{itemize}
	\item Left/right division of two skew polynomials of degree at most $n$
	\item Generalized operator / remainder evaluation of a skew polynomial of degree at most $n$ at $n$ elements from $\Fqm$
	\item Computation of the minimal polynomials $\minpolyRem{\set{B}}$ and $\minpolyOp{\set{B}}{\vec{a}}$ for $|\set{B}|\leq n$ w.r.t. the remainder and generalized operator evaluation, respectively
	\item Computation of the \ac{LCLM} (see~\cite[Theorem~3.2.7]{caruso2017new})
\end{itemize}

\subsection{Skew Kötter--Nielsen--H\o holdt Interpolation}

We now consider the skew \ac{KNH} interpolation from~\cite{liu2014kotter}, which is the skew polynomial analogue of the \ac{KNH} interpolation over ordinary polynomial rings in~\cite{wang2005kotter}.
Note, that due to the isomorphism between $\SkewPolyring$ and the ring of linearized polynomials for $\aut$ being the Frobenius automorphism and $\der=0$ (zero derivations), the~\ac{KNH} variant over linearized polynomial rings in~\cite{xie2011general} can be seen as a special case of~\cite{liu2014kotter}.

As input to our problem, we consider the $n$ $\Fqm$-linear skew vector evaluation maps\footnote{In~\cite{wang2005kotter,liu2014kotter,xie2011general} linear functionals are defined for each interpolation point. 
Here we give a different definition based on skew polynomial vectors which is equivalent to the definition based on linear functionals if $\SkewPolyring^{\intParam+1}$ is considered as a vector space over $\Fqm$.} $\vecEvNoInput{i}$:
\begin{equation}\label{eq:skew_eval_maps}
   \vecEvNoInput{i} : \SkewPolyring^{\intParam+1} \rightarrow \Fqm \nonumber
\end{equation}
where $n$ is the number of interpolation constraints and $\intParam$ is an interpolation parameter.

Later on, we will specify particular mappings $\vecEvNoInput{i}$ w.r.t. the generalized operator and the remainder evaluation (see Section~\ref{sec:coding_applications}).
For each skew vector evaluation map $\vecEvNoInput{i}$ we define the kernels
\begin{equation}
	\module{K}_i\defeq\{\vec{Q}\in\SkewPolyring^{\intParam+1}:\vecEv{i}{\vec{Q}}=0\},\quad\forall i=0,\dots,n-1.
\end{equation}
For $0\leq i\leq n-1$ the intersection $\bar{\module{K}}_i\defeq\module{K}_0\,\cap\,\module{K}_1\,\cap\,\dots\,\cap\,\module{K}_i$ contains all vectors from $\SkewPolyring^{\intParam+1}$ that are mapped to zero under $\vecEvNoInput{0},\vecEvNoInput{1},\dots,\vecEvNoInput{i}$, i.e.
\begin{equation}
 \bar{\module{K}}_i=\{\vec{Q}\in\SkewPolyring^{\intParam+1}:\vecEv{j}{\vec{Q}}=0,\forall j=0,\dots,i\}.
\end{equation}

Under the assumption that the $\bar{\module{K}}_i$ are left $\SkewPolyring$-submodules for all $i=0,\dots,n-1$ (see~\cite{liu2014kotter}) we can state the general skew polynomial vector interpolation problem.

\begin{problem}[General Vector Interpolation Problem]\label{prob:generalIntProblem}
Given the integer $\intParam\in\mathbb{Z}_+$, a set of $\Fqm$-linear vector evaluation maps $\set{E}=\{\vecEvNoInput{0},\dots,\vecEvNoInput{n-1}\}$ and a vector $\vec{w}\in\mathbb{Z}_+^{\intParam+1}$ compute a $\vec{w}$-ordered weak-Popov Basis for the left $\SkewPolyring$-module
\begin{equation}\label{eq:int_module}
	\bar{\module{K}}_{n-1}=\{\vec{b}\in\SkewPolyring^{\intParam+1}:\vecEv{i}{\vec{b}}=0,\forall i=0,\dots,n-1\}.
\end{equation}
\end{problem}

Problem~\ref{prob:generalIntProblem} can be solved using a slightly modified variant of the multivariate skew \ac{KNH} interpolation from~\cite{liu2014kotter}.
Since the solution of Problem~\ref{prob:generalIntProblem} is a $\w$-ordered weak Popov basis for the interpolation module $\bar{\module{K}}_{n-1}$ instead of a single minimal polynomial vector, we modified the output of~\cite[Algorithm~1]{liu2014kotter} such that it returns a whole basis for the interpolation module $\bar{\module{K}}_{n-1}$. A similar approach was used in~\cite{bartz2017algebraic} to construct a basis for the interpolation module over linearized polynomial rings.

The modified multivariate skew \ac{KNH} interpolation is summarized in Algorithm~\ref{alg:skewMultVarKNH}.

\begin{algorithm}
	\caption{Modified Skew \ac{KNH} Interpolation}\label{alg:skewMultVarKNH}
	\SetKwInOut{Input}{Input}\SetKwInOut{Output}{Output}
	\Input{
		A set $\{\vecEvNoInput{0},\vecEvNoInput{1},\dots,\vecEvNoInput{n-1}\}$ of vector evaluation maps \\ 
		A ``weighting'' vector $\vec{w}=(w_0, w_1, \dots, w_{\intParam})\in\posInt^{\intParam+1}$
	}
	\Output{A $\vec{w}$-ordered weak-Popov Basis $\mat{B}\in\SkewPolyring^{(\intParam+1)\times(\intParam+1)}$ for $\intModule{n-1}$
	}
	\textbf{Initialize:}
	$\mat{B}=\mat{I}_{\intParam+1}\in\SkewPolyring^{(\intParam+1)\times(\intParam+1)}$\\

	\BlankLine
	\For{$i\leftarrow 0$ \KwTo $n-1$}{
		\For{$j\leftarrow 0$ \KwTo $\intParam$}{
			$\Delta_{j}\gets \vecEv{i}{\vec{b}_j}$ \label{alg1:functional}
		}
		$\set{J}\gets \{j:\Delta_{j}\neq 0\}$ \\
		\If{$\set{J}\neq \emptyset$}{
			$j^{*}\gets \min_{j\in \set{J}}\{\arg\min_{j\in \set{J}}\{\deg_\vec{w}(\vec{b}_{j})\}\}$ \label{alg1:minargmin} \\
			$\vec{b}^{*} \gets \vec{b}_{j^{*}}$ \\
			\For{$j\in \set{J}$}{
				\If{$j=j^{*}$}{
					$\vec{b}_{j} \gets \left(x-\frac{\vecEv{i}{x \vec{b}^{*}}}{\Delta_{j^{*}}}\right) \vec{b}^{*}$ \label{alg1:degreeinc} \tcc*[r]{degree-increasing step}					
				}
				\Else{	
					$\vec{b}_{j} \gets \vec{b}_{j} - \frac{\Delta_{j}}{\Delta_{j^{*}}}\vec{b}^{*}$ \label{alg1:crosseval} \tcc*[r]{cross-evaluation step}
				}
			}
		}
	}
	\Return{$\mat{B}$}
\end{algorithm}

Note, that $\min_{j\in \set{J}}\{\arg\min_{j\in \set{J}}\{\deg_\vec{w}(\vec{b}_{j})\}\}$ in Line~\ref{alg1:minargmin} returns the smallest index $j\in\set{J}$ to break ties, i.e. the index $j$ of the minimal vector $\b_j$ w.r.t. $\wOrder$ for which $\Delta_j \neq 0$ (see~\eqref{eq:top_order}).

In each iteration of Algorithm~\ref{alg:skewMultVarKNH} (and so~\cite[Algorithm~1]{liu2014kotter}) there are three possible update steps:
\begin{enumerate}
 \item \emph{No update}: The vector $\b_j$ is not updated if $\b_j$ is in the kernel $\bar{\module{K}}_i$ already, i.e. if $\Delta_j=\vecEv{i}{\vec{b}_j}=0$.

 \item \emph{Cross-evaluation} (or \emph{order-preserving}~\cite{liu2014kotter}) update: For any $\b_j$ that is not minimal w.r.t. $\wOrder$ (i.e. $j\neq j^*$) the cross-evaluation update (Line~\ref{alg1:crosseval}) is performed such that
 \begin{align*}
  \vecEvNoInput{i}\left(\vec{b}_{j} - \frac{\Delta_{j}}{\Delta_{j^{*}}}\vec{b}^{*}\right)
  =\vecEv{i}{\vec{b}_{j}} - \frac{\vecEv{i}{\vec{b}_{j}}}{\vecEv{i}{\vec{b}_{j^*}}}\vecEv{i}{\vec{b}_{j^*}}
  =0.
 \end{align*}
 Note, that the ($\w$-weighted) degree of $\b_j$ is not increased by this update.
 
 \item \emph{Degree-increasing} (or \emph{order-increasing}~\cite{liu2014kotter}) update: For the minimal vector $\b_{j^*}\defeq \b^*$ w.r.t. $\wOrder$ the degree-increasing update (Line~\ref{alg1:degreeinc}) is performed such that
 \begin{align*}
  \vecEvNoInput{i}\left(\left(x-\frac{\vecEv{i}{x \vec{b}^{*}}}{\Delta_{j^{*}}}\right) \vec{b}^{*}\right)
  =\vecEv{i}{x\b^{*}}-\frac{\vecEv{i}{x \vec{b}^{*}}}{\vecEv{i}{\vec{b}^{*}}}\vecEv{i}{ \vec{b}^{*}}
  =0.
 \end{align*}
 The ($\w$-weighted) degree of $\b^*$ is increased by one in this case. 
\end{enumerate}
The different update steps are illustrated in~\cite[Figure~1]{liu2014kotter}.
Define the sets 
\begin{equation}
    \set{S}_j = \{\Q\in\SkewPolyring^{\intParam+1}:\ind{\Q}{\w}=j\}\cup\{\0\}
\end{equation}
and
\begin{equation}
    \set{T}_{i,j} = \bar{\module{K}}_i\cap\set{S}_j
\end{equation}
for all $i=0,\dots,n-1$ and $j=0,\dots,\intParam$.
Note, that $\set{S}_j\cap\set{S}_{j'}=\{\0\}$ for all $0\leq j,j'\leq \intParam$.

The following result from~\cite{liu2014kotter} is fundamental for proving the correctness of Algorithm~\ref{alg:skewMultVarKNH}.

\begin{theorem}[{\cite[Theorem~5]{liu2014kotter}}]\label{thm:KNHminimality}
 After each iteration $i$ of Algorithm~\ref{alg:skewMultVarKNH}, the updated $\b_j$ is a minimum w.r.t. $\wOrder$ in $\set{T}_{i,j}$ for all $j=0,\dots,\intParam$.
\end{theorem}

In other words, after the $i$-th iteration each $\b_j$ has $\ind{\b}{\w}=j$ and the minimal $\w$-weighted degree among all vectors in $\set{T}_{i,j}$.
Therefore, after the $i$-th iteration, the matrix $\B$ is a $\w$-ordered weak Popov basis for $\bar{\module{K}}_i$.

\begin{lemma}[Correctness of Algorithm~\ref{alg:skewMultVarKNH}]\label{lem:correctness_skewMultVarKNH}
 Algorithm~\ref{alg:skewMultVarKNH} is correct and provides a solution to the general vector interpolation problem in Problem~\ref{prob:generalIntProblem}.
\end{lemma}

\begin{proof}
 The update steps of Algorithm~\ref{alg:skewMultVarKNH} and~\cite[Algorithm~1]{liu2014kotter} are equivalent and therefore we have by~\cite[Theorem~5]{liu2014kotter} that after the $i$-th iteration each $\b_j\in\bar{\module{K}}_i$.
 We now will show that after the $i$-th iteration of Algorithm~\ref{alg:skewMultVarKNH} the matrix $\B$ is a $\w$-ordered weak Popov basis for $\bar{\module{K}}_i$.
 By Theorem~\ref{thm:KNHminimality} (\cite[Theorem~5]{liu2014kotter}) each $\b_j$ has the minimal $\w$-weighted degree among all polynomials in $\set{T}_{i,j}$, which implies that the $\w$-pivot indices of $\b_0,\dots,\b_\intParam$ are increasing and distinct.
 Now assume the there exists a vector $\p\in\bar{\module{K}}_i$ that can not be represented by a $\SkewPolyring$-linear combination of the form
 \begin{equation*}
  \p=\sum_{j=0}^{\intParam}a_j\b_j
 \end{equation*}
 for some $a_j\in\SkewPolyring$.
 Then we must have that $\p$ can be written as
 \begin{equation}
     \p=\r + \sum_{j=0}^{\intParam}a_j\b_j
 \end{equation}
 where $\deg_\w(\r)<\min_j\{\deg_\w(\b_j)\}$.
 This contradicts that $\b_j$ is a minimum w.r.t. $\wOrder$ in $\set{T}_{i,j}$ since $\ind{\r}{\w}\in\{0,\dots,\intParam\}$.
 Therefore we conclude that after the $i$-th iteration $\B$ is a $\w$-ordered weak Popov basis for $\bar{\module{K}}_i$.
\end{proof}

\begin{proposition}[Computational Complexity of Algorithm~\ref{alg:skewMultVarKNH}]\label{prop:complSkewMultVarKNH}
 The complexity of Algorithm~\ref{alg:skewMultVarKNH} is dominated by the complexity of:
 \begin{itemize}
    \item $\OCompl{\intParam n}$ evaluation maps $\vecEvNoInput{i}$ applied to a vector from $\SkewPolyring^{\intParam+1}_{\leq n}$,
    
    \item $n$ multiplications of a monic degree-1 skew polynomial with a vector from $\SkewPolyring^{\intParam+1}_{\leq n}$ (degree-increasing step),
    
    \item $\OCompl{\intParam n}$ multiplications of an element from $\Fqm$ with a vector from $\SkewPolyring^{\intParam+1}_{\leq n}$ (cross-evaluation step).
 \end{itemize}
\end{proposition}

\begin{proof}
 In each of the $n$ iterations we have:
 \begin{itemize}
  \item $\intParam+2$ evaluation maps $\vecEvNoInput{i}$ applied to a vector from $\SkewPolyring^{\intParam+1}_{\leq n}$ (Line~\ref{alg1:functional}),
  
  \item one product of a skew polynomial of degree $1$ with a vector from $\SkewPolyring^{\intParam+1}_{\leq n}$ (degree-increasing step in Line~\ref{alg1:degreeinc}),

  \item $\intOrder$ multiplications of an element from $\Fqm$ with a vector from $\SkewPolyring^{\intParam+1}_{\leq n}$ (cross-evaluation step in Line~\ref{alg1:crosseval}),

  \item $\intParam+1$ inversions/divisions in $\Fqm$.
 \end{itemize}
 \qed
\end{proof}
\section{Fast Kötter--Nielsen--H\o holdt Interpolation over Skew Polynomial Rings}\label{sec:fast_skew_knh}

In~\cite{nielsen2014fast} a fast \ac{DaC} variant of the Kötter interpolation for the Guruswami--Sudan decoder for Reed--Solomon codes was presented.
We now use ideas from~\cite{nielsen2014fast} to speed up the skew \ac{KNH} interpolation from~\cite{liu2014kotter}. 
The main idea is to sub-divide Problem~\ref{prob:generalIntProblem} into a tree of successively smaller problems. 
Each leaf is identified with a linear functional, and the updates done here are represented as skew polynomial matrices. 
The inner-nodes of the tree combine updates using matrix multiplication. 
The entire cost of the algorithm hinges on the fact that at any node of the tree, we need only know the intermediate basis up to its image on the linear functionals of the subtree of that node.

In the following, we describe the general framework for the fast skew~\ac{KNH} interpolation algorithm which we will then discuss in Section~\ref{sec:coding_applications} w.r.t. to particular operator and remainder vector evaluation maps.

The operations performed on the basis $\mat{B}$ in the inner loop of the $i$-th iteration of Algorithm~\ref{alg:skewMultVarKNH} can be represented by the matrix
\begin{equation}\label{eq:opMatrix}
\arraycolsep=5pt
\mat{U}=
\left(
\begin{array}{ccc|c|ccc}
1 & & & -\frac{\Delta_0}{\Delta_{j^{*}}} & & &
\\
& \ddots & & \vdots & & &
\\
& & 1 & -\frac{\Delta_{j^{*}-1}}{\Delta_{j^{*}}} & & &
\\
& &	& \left(x-\frac{\vecEv{i}{x \vec{b}_{j^{*}}}}{\Delta_{j^{*}}}\right) & & &
\\
& & & -\frac{\Delta_{j^{*}+1}}{\Delta_{j^{*}}} & 1 & &
\\
& & & \vdots &  & \ddots &
\\
& & & -\frac{\Delta_{\intParam}}{\Delta_{j^{*}}} &  & & 1
\end{array}
\right)
\end{equation}
such that after the $i$-th iteration we obtain the basis $\mat{U}\mat{B}$ for $\bar{\module{K}}_i$.
Note, that if $\Delta_j=0$ no update on the row $\vec{b}_j$ is performed since $\Delta_j/\Delta_{j^{*}}=0$. 

\subsection{Divide-and-Conquer Skew Kötter Interpolation}

To describe the following algorithms we introduce some notations. $\vec{M}_{[i,j]} \in \SkewPolyring^{\intParam+1}$ denotes a polynomial vector that is dependent on the index set $\{i,i+1,\ldots,j-1,j\}$ with $j\geq i$ and $\vec{M}_{[i,i]} = \vec{M}_i$.
The (ordered) set $\mathcal{M}$ is globally available for all algorithms and is defined as
\begin{equation}
\set{M}=(\vec{M}_{[0,n-1]},\vec{M}_{[0,\lfloor n/2\rfloor-1]},\vec{M}_{[\lfloor n/2\rfloor,n-1]},\ldots,\vec{M}_{0},\vec{M}_{1},\ldots,\vec{M}_{n-1})\subseteq\SkewPolyring^{\intParam+1}
\end{equation}
for an integer $n\in\posInt$.
For an ordered set (or tuple) of evaluation maps $\set{E} = (\vecEvNoInput{0},\ldots,\vecEvNoInput{n-1})$ we use a similar notation to access an ordered subset of $\set{E}$ as follows
\begin{equation}
\set{E}_{[i,j]}=(\vecEvNoInput{i},\dots,\vecEvNoInput{j}).
\end{equation}
Depending on the considered interpolation problem, we will later on define the polynomial vectors $\vec{M}_{[i,j]}$ to contain minimal polynomials that depend on the interpolation points corresponding to the vector evaluation maps in $\set{E}_{[i,j]}$.
In the general interpolation problem (Problem~\ref{prob:generalIntProblem}) we only consider sets of evaluation maps whereas here we consider \emph{ordered} sets since we need the notation within the \ac{DaC} algorithm to build up the tree.

In order to describe a general framework for the fast skew~\ac{KNH} interpolation, we need the following assumption. In Section~\ref{sec:coding_applications} we show that this assumption holds for specific coding applications.

\begin{assumption}\label{ass:mod_vectors}
Let $\set{E}=\{\vecEvNoInput{0},\ldots,\vecEvNoInput{n-1}\}$ be a set of linear functionals as defined in~\eqref{eq:skew_eval_maps} and let $\set{E}_{[i,j]}=\{\vecEvNoInput{i},\ldots,\vecEvNoInput{j}\}\subseteq\set{E}$. 
We assume that for all $\Q\in\SkewPolyring^{\intParam+1}$ and $0\leq i\leq j\leq n-1$ the skew polynomial vector $\vecMinpolyNoX{[i,j]}\in\set{M}$ (which contains minimal skew polynomials that depend on $\set{E}_{[i,j]}$) satisfies
 \begin{equation}
  \vecEv{l}{\vec{Q}}=\vecEv{l}{\vec{Q}\modr\vecMinpolyNoX{[i,j]}},
  \quad\forall l=i,\dots,j.
 \end{equation}
\end{assumption}

\begin{algorithm}
	\caption{\textsf{SkewInterpolatePoint}}\label{alg:skewIntPoint}
	\SetKwInOut{Input}{Input}\SetKwInOut{Output}{Output}
	\Input{A skew vector evaluation map $\vecEvNoInput{i}$ 
	,\\ 
	$\mat{B} \in \SkewPolyring^{(\intParam+1)\times(\intParam+1)}$\\
	$\vec{d}=(d_0,d_1,\ldots,d_{s}) \in \posInt^{\intParam+1}$ s.t. $d_j=\deg_\vec{w}(\vec{b}_j)$ for all $j=0,\dots,\intParam$.
	}
	\BlankLine
	\Output{$\mat{T}\in\SkewPolyring^{(\intParam+1)\times(\intParam+1)}$ s.t. the rows of $\mat{\hat{B}}\defeq\mat{TB}$ is a $\vec{w}$-ordered weak-Popov Basis for $\spannedBy{\mat{B}}\cap\module{K}_i$,\\ 
	$\vec{\hat{d}}=(\hat{d}_0,\hat{d}_1,\ldots,\hat{d}_{s}) \in \posInt^{\intParam+1}$ s.t. $\hat{d}_j=\deg_\vec{w}(\vec{\hat{b}}_j)$ for all $j=0,\dots,\intParam$.
	}

	\BlankLine
	$\vec{\hat{d}} \gets \vec{{d}}$ \\
	\For{$j\leftarrow 0$ \KwTo $\intParam$}{
		$\Delta_{j}\gets \vecEv{i}{\vec{b}_j}$
	}
	$\set{J}\gets \{j:d_{j}\neq 0\}$ \\
	$\mat{T}\gets\mat{I}_{\intParam+1}\in\SkewPolyring^{(\intParam+1)\times(\intParam+1)}$ \\
	\If{$\set{J}\neq \emptyset$}{
			$j^{*}\gets\min_{l\in \set{J}}\{\arg\min_{l\in \set{J}}\{d_l\}\}$ \\
			$\mat{T}\gets\mat{U}$ where $\mat{U}$ is as in~\eqref{eq:opMatrix}\\
			$\hat{d}_{j^{*}}\gets\hat{d}_{j^{*}}+1$
	}
	\Return{$(\mat{T}, \vec{\hat{d}})$}
\end{algorithm}

\begin{lemma}[Correctness of Algorithm~\ref{alg:skewIntPoint}]
 \textsf{SkewInterpolatePoint} in Algorithm~\ref{alg:skewIntPoint} is correct.
\end{lemma}

\begin{proof}
The columns of $\U$ except for the $j^*$-th column correspond to the non-minimal rows of $\mat{B}$. 
The cross-evaluation step of the $j$-th non-minimal row of $\B$ in Line~\ref{alg1:crosseval} of Algorithm~\ref{alg:skewMultVarKNH} is performed by the entries in the $j$-th row and $j^{*}$-th column of $\mat{U}$. 
The entry in the $j^{*}$-th row and the $j^{*}$-th column of $\mat{U}$ corresponds to the degree-increasing step in Line~\ref{alg1:degreeinc} in Algorithm~\ref{alg:skewMultVarKNH}. 
Hence, the algorithm outputs a matrix $\mat{T}$ such that all rows of $\mat{T}\mat{B}$ are mapped to zero under $\vecEvNoInput{i}$.
 \qed
\end{proof}

Equipped with the routine \textsf{SkewInterpolatePoint} in Algorithm~\ref{alg:skewIntPoint} to solve the basic step we can now derive a~\ac{DaC} variant of the skew~\ac{KNH} interpolation in Algorithm~\ref{alg:skewMultVarKNH}.

\begin{algorithm}
	\caption{\textsf{SkewInterpolateTree}}\label{alg:skewIntTree}
	\SetKwInOut{Input}{Input}\SetKwInOut{Output}{Output}
	\Input{
	Skew vector evaluation maps $\set{E}_{[i_1,i_2]}=(\vecEvNoInput{i_1},\dots,\vecEvNoInput{i_2})$ \\
	$\mat{B}\in\SkewPolyring^{(\intParam+1)\times(\intParam+1)}$, \\
	$\vec{d}=(d_0,d_1,\ldots,d_{s}) \in \posInt^{\intParam+1}$ s.t. $d_j=\deg_\vec{w}(\vec{b}_j)$ for all $j=0,\dots,\intParam$.
	}
	\BlankLine

	\Output{
	A matrix $\mat{T}\in\SkewPolyring^{(\intParam+1)\times(\intParam+1)}$ s.t. $\mat{\hat{B}}\defeq\mat{TB}$ is a $\vec{w}$-ordered weak-Popov Basis for $\spannedBy{\mat{B}}\cap\module{K}_{i_1}\cap\dots\cap\module{K}_{i_2}$, \\ 
	$\vec{\hat{d}}=(\hat{d}_0,\hat{d}_1,\ldots,\hat{d}_{s}) \in \posInt^{\intParam+1}$ s.t. $\hat{d}_j=\deg_\vec{w}(\vec{\hat{b}}_j)$ for all $j=0,\dots,\intParam$.
	}

	\BlankLine
	\If{$i_1=i_2$}{
		\Return{$\textsf{SkewInterpolatePoint}(\vecEvNoInput{i_1}, {\mat{B}}, \vec{d})$}
	}
	\Else{
		$z\gets\left\lfloor\frac{i_1+i_2}{2}\right\rfloor$ \\
		${\mat{B}}_1\gets{\mat{B}}\modr\vecMinpolyNoX{[i_1,z]}$ \label{line:VecMod1} \\
		$(\mat{T}_1, \vec{d}_1)\gets \textsf{SkewInterpolateTree}(\set{E}_{[i_1,z]}, {\mat{B}}_1, \vec{d})$ \\
		${\mat{B}}_2\gets\mat{T}_1{\mat{B}}\modr \vecMinpolyNoX{[z+1,i_2]}$ \label{line:VecMod2} \\ 
		$(\mat{T}_2, \vec{d}_2)\gets \textsf{SkewInterpolateTree}(\set{E}_{[z+1,i_2]}, {\mat{B}}_2, \vec{d}_1)$ \\ 
		\Return{$(\mat{T}=\mat{T}_2\mat{T}_1, \vec{\hat{d}}=\vec{d}_2)$}\label{line:T2T1}
	}
\end{algorithm}

\begin{lemma}[Correctness of Algorithm~\ref{alg:skewIntTree}]
 Under Assumption~\ref{ass:mod_vectors}, \textsf{SkewInterpolateTree} in Algorithm~\ref{alg:skewIntTree} is correct.
\end{lemma}

\begin{proof}
 The correctness of Algorithm~\ref{alg:skewIntTree} follows directly from the definition of the matrix $\mat{U}$ in~\eqref{eq:opMatrix} and Assumption~\ref{ass:mod_vectors}.
 \qed
\end{proof}

\subsection{Precomputing Minimal Polynomial Vectors}

We now present a generic procedure to pre-compute the set $\set{M}$ containing the minimal polynomial vectors $\vecMinpolyNoX{[i,j]}$ required in Algorithm~\ref{alg:skewIntTree} efficiently.
We consider minimal polynomials such as the generalized operator and the remainder evaluation which can be constructed by means of the \ac{LCLM} of polynomial sequences (see~\eqref{eq:def_minpoly_gen_op} and~\eqref{eq:def_minpoly_rem}) we use the ideas from~\cite[Theorem~3.2.7]{caruso2017new} to obtain the efficient procedure described in Algorithm~\ref{alg:precomputeMinVectors}. 
The \ac{DaC} structure of the algorithm is illustrated in Figure~\ref{fig:minpoly_vec_tree} for an example of $\kappa=4$. The initial minimal polynomial vectors $\vecMinpolyNoX{0},\vecMinpolyNoX{1},\ldots,\vecMinpolyNoX{\kappa-1}$ from which all other minimal polynomials are computed via the \ac{LCLM}, are computed depending on the application. In Section~\ref{sec:coding_applications} two cases are given, for the general operator evaluation as in~\eqref{eq:def_gen_op_minpoly_vecs} and for the remainder evaluation as in~\eqref{eq:def_rem_minpoly_vecs}.

\begin{algorithm}
	\caption{\textsf{PreComputeMinVectorsTree}}\label{alg:precomputeMinVectors}
	\SetKwInOut{Input}{Input}\SetKwInOut{Output}{Output}
	\Input{Upper and lower index bound $a\in\mathbb{Z}_+$ and $b\in\mathbb{Z}_+$ with $b\geq a$ \\ 
		Minimal polynomial vectors $\vecMinpolyNoX{a},\vecMinpolyNoX{a+1},\ldots,\vecMinpolyNoX{b} \in \SkewPolyring^{\intParam+1}$
	}
	\Output{An (ordered) set $(\vec{M}_{[a,b]},\vec{M}_{[a,\lfloor (b-1)/2\rfloor-1]},\vec{M}_{[\lfloor (b-1)/2\rfloor,b]},\ldots,\vec{M}_{[a,a]},\vec{M}_{[a+1,a+1]},\ldots,\vec{M}_{[b,b]}) \subset\SkewPolyring^{\intParam+1}$}

	\BlankLine
	
	\If{$a = b$} {
		\Return $\{\vecMinpolyNoX{a}\}$
	}
	\Else{
		$\delta \gets \lfloor \frac{b-a+1}{2} \rfloor$\\
		$\set{M}_{1} \gets \text{PreComputeMinVectorsTree}(a,a+\delta-1) $ \\
		$\set{M}_{2} \gets \text{PreComputeMinVectorsTree}(a+\delta,b) $ \\
		$\vecMinpolyNoX{[a,b]} \gets \lclm{\vecMinpolyNoX{[a,a+\delta-1]},\vecMinpolyNoX{[a+\delta,b]}}{}$ \tcc*[r]{with $\vecMinpolyNoX{[a,a+\delta-1]} \in \set{M}_1$ and $\vecMinpolyNoX{[a+\delta,b]} \in \set{M}_2$}
		\Return $ \set{M}_{1} \cup \set{M}_{2} \cup\{\vecMinpolyNoX{[a,b]}\}$
		}		
\end{algorithm}

\begin{lemma}[Correctness of Algorithm~\ref{alg:precomputeMinVectors}]
 Algorithm~\ref{alg:precomputeMinVectors} is correct.
\end{lemma}

\begin{proof}
 The correctness of Algorithm~\ref{alg:precomputeMinVectors} follows directly from~\cite[Theorem~3.2.7]{caruso2017new}.
 The algorithm proceeds in a recursive manner and splits the size of the set of considered minimal polynomials in half. When sets consist only of one element, $\vecMinpolyNoX{[a,a]}$ are computed, using the generalized operator or remainder evaluation (see \eqref{eq:def_minpoly_gen_op} and \eqref{eq:def_minpoly_rem}). The sets of minimal polynomials of larger size are then obtained by merging the smaller sets of minimal polynomials using the relation $\vecMinpolyNoX{[a,b]} = \lclm{\vecMinpolyNoX{[a,a+\delta-1]},\vecMinpolyNoX{[a+\delta,b]}}{}$ with $\delta = \lfloor \frac{b-a+1}{2} \rfloor$, also illustrated in Figure~\ref{fig:minpoly_vec_tree}.
 \qed
\end{proof}

\begin{figure}[ht!]
 \centering
 \tikzset{
  solid node/.style={circle,draw,inner sep=1.2,fill=black},
  hollow node/.style={circle,draw,inner sep=1.2},
}

\begin{tikzpicture}[font=\footnotesize]
  \tikzset{
    level 1/.style={level distance=15mm,sibling distance=45mm},
    level 2/.style={level distance=15mm,sibling distance=30mm},
    level 3/.style={level distance=15mm,sibling distance=15mm},
    level 4/.style={level distance=15mm,sibling distance=10mm},
  }

  \node[hollow node,label=above:{$\vecMinpolyNoX{[0,3]}=\lclm{\vecMinpolyNoX{[0,1]},\vecMinpolyNoX{[2,3]}}{}$}] (root) {}
    child{node[solid node,label=left:{$\vecMinpolyNoX{[0,1]}=\lclm{\vecMinpolyNoX{[0,0]},\vecMinpolyNoX{[1,1]}}{}$}]{}
      child{node(l1)[solid node,label=below:{$\vecMinpolyNoX{[0,0]}=\vecMinpolyNoX{0}$},label=below:{}]{}
        edge from parent node[left]{}
      }
      child{node(l2)[solid node,label=below:{$\vecMinpolyNoX{[1,1]}=\vecMinpolyNoX{1}$},label=below:{}]{}
        edge from parent node[right]{}
      }
      edge from parent node[left,xshift=-10]{}
    }
    child{node[solid node,label=right:{$\vecMinpolyNoX{[2,3]}=\lclm{\vecMinpolyNoX{[2,2]},\vecMinpolyNoX{[3,3]}}{}$}]{}
      child{node(r1)[solid node,label=below:{$\vecMinpolyNoX{[2,2]}=\vecMinpolyNoX{2}$}]{}
        edge from parent node[left]{}
      }
      child{node(r2)[solid node,label=below:{$\vecMinpolyNoX{[3,3]}=\vecMinpolyNoX{3}$}]{}
        edge from parent node[right]{}
      }
      edge from parent node[right,xshift=10]{}
    }
  ;
  \node[above of = root]{};
\end{tikzpicture}
 \caption{Illustration of the computation tree of Algorithm~\ref{alg:precomputeMinVectors} to precompute all minimal polynomial vectors in the set $\mathcal{M}$ for $\kappa=4$.}
 \label{fig:minpoly_vec_tree}
\end{figure}
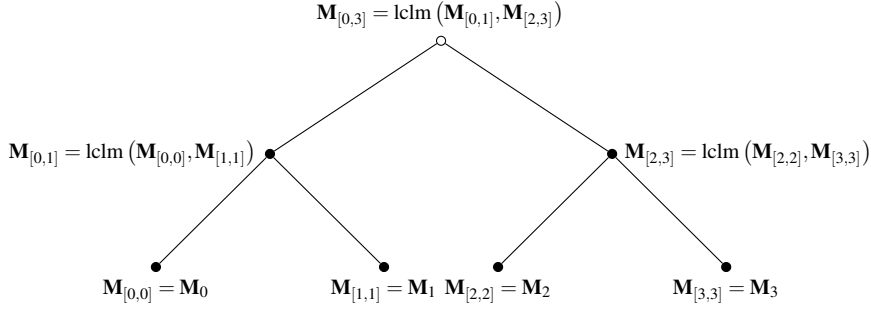

\section{Application to Coding Problems}\label{sec:coding_applications}

In this section we apply the fast skew \ac{KNH} interpolation described in Section~\ref{sec:fast_skew_knh} to coding problems in the rank, sum-rank and skew metric.
In particular, we consider evaluation codes over constructed using $\SkewPolyringZeroDer$, i.e. the skew polynomial ring with zero derivations only.
The results can be generalized to the $\SkewPolyring$ case, except for the complexity analysis.

\subsection{Interpolation-Based Decoding of Interleaved Gabidulin Codes}

Interleaved Gabidulin codes are rank-metric codes that are obtained by the Cartesian product of ordinary Gabidulin codes~\cite{Gabidulin_TheoryOfCodes_1985} and which allow for decoding beyond half the minimum rank distance (see~\cite{Loidreau_Overbeck_Interleaved_2006,overbeck2006decoding}).
In~\cite{Loidreau_Overbeck_Interleaved_2006}, a probabilistic unique decoder that is able to correct errors beyond the unique decoding radius with high probability, was presented.
A Berlekamp--Massey-like decoding algorithm for interleaved Gabidulin codes was presented in~\cite{sidorenko2010decoding,sidorenko2011skew}.

A Welch--Berlekamp-like interpolation-based decoding scheme, that can be either used as a (not necessarily polynomial-time) list decoding algorithm or as a probabilistic unique decoding algorithm was presented by Wachter-Zeh and Zeh in~\cite{WachterzehZeh-ListUniqueErrorErasureInterpolationInterleavedGabidulin_DCC2014}.
The algorithm consists of an \emph{interpolation step} and a \emph{root-finding step}.
The list size (and so the probability of obtaining a unique solution) can be optimized by using a minimal Gröbner basis for the left module containing the solutions of the interpolation problem rather than a single solution of the interpolation problem for the root-finding step (see~\cite{bartz2018efficient,bartz2017algebraic}).

The interpolation step (including the computation of the minimal Gröbner basis) can be solved by the linearized or skew variants of the multivariate \ac{KNH} interpolation~\cite{xie2011general,liu2014kotter} in $\oh{\intOrder^2n^2}$ operations in $\Fqm$, where $\intOrder$ is the \emph{interleaving order} and $n$ is the length of the code.
The overall computational complexity of the algorithm in~\cite{WachterzehZeh-ListUniqueErrorErasureInterpolationInterleavedGabidulin_DCC2014} is then in the order of $\oh{\intOrder^2n^2}$ operations in $\Fqm$.

Recently, efficient algorithms for solving the interpolation step and the root-finding step that are based on the computation of minimal approximant bases over skew polynomial rings were presented in~\cite{bartz2021fast}. 
The algorithms reduce the overall decoding complexity of the Wachter-Zeh and Zeh decoder~\cite{WachterzehZeh-ListUniqueErrorErasureInterpolationInterleavedGabidulin_DCC2014} to $\softoh{\intOrder^\matmult \OMul{n}}\subseteq\softoh{\intOrder^\matmult n^{1.635}}$.

In this section we show how the fast skew \ac{KNH} interpolation algorithm from Section~\ref{sec:fast_skew_knh} can be used to accomplish the interpolation step in the Wachter-Zeh and Zeh decoder in $\softoh{\intOrder^\matmult \OMul{n}}\subseteq\softoh{\intOrder^\matmult n^{1.635}}$ operations in $\Fqm$. 
The obtained improved computational complexity coincides with the computational complexity of the minimal approximant bases variant from~\cite{bartz2021fast}.
The results can be also used to speed up the interpolation-based decoder for interleaved subspace codes in~\cite{bartz2018efficient}.

\subsubsection{Codes in the Rank Metric}

The \emph{rank weight} of a vector $\vec{x}\in\Fqm^n$ is defined as (see~\cite{Gabidulin_TheoryOfCodes_1985})
\begin{equation}\label{eq:def_rank_weight}
	\RankWeight(\vec{x})\defeq\rk_q(\vec{x})=\rk_q(\mat{X})
\end{equation}
where $\mat{X}\in\Fq^{m\times n}$ is the corresponding expanded matrix of $\vec{x}$ over $\Fq$. 
Notice, that in general we have that $\RankWeight(\vec{x})\leq\HammingWeight(\vec{x})$ for any $\vec{x}\in\Fqm^n$ (see~\cite{martinez2016similarities}).
The \emph{rank distance} between two vectors $\vec{x},\vec{y}\in\Fqm^n$ is then defined as the rank of their difference, i.e. as
\begin{equation}\label{eq:def_rank_distance}
 \RankDist(\vec{x},\vec{y})\defeq\RankWeight(\vec{x}-\vec{y}).
\end{equation}
As a channel model we consider the rank error channel 
\begin{equation}
	\vec{r}=\vec{c}+\vec{e}
\end{equation}
where $\RankWeight(\vec{e})=t$.

\begin{definition}[Interleaved Gabidulin Code~\cite{Loidreau_Overbeck_Interleaved_2006,overbeck2006decoding}]
 Let $\aut$ be the Frobenius automorphism $\aut(\cdot)=\cdot^q$ of $\Fqm$.
 Let $\bm{\beta}=(\beta_0,\beta_1,\dots,\beta_{n-1})\in\Fqm^{n}$ contain $\Fq$-linearly independent elements from $\Fqm$.
 An $\intOrder$-interleaved Gabidulin code of length $n\leq m$ and dimension $k$ is defined as
 \begin{equation}
 	\IntGab{\bm{\beta},\intOrder;n,k}
 	=\left\{
 	\begin{pmatrix}
 	 \opev{f^{(1)}}{\bm{\beta}}{1}
 	 \\
 	 \opev{f^{(2)}}{\bm{\beta}}{1} 
 	 \\
 	 \vdots 
 	 \\
 	 \opev{f^{(\intOrder)}}{\bm{\beta}}{1}
 	\end{pmatrix}
 	: f^{(j)}\in\SkewPolyringZeroDer_{<k},\forall j\in\intervallincl{1}{\intOrder}
 	\right\}.
 \end{equation}
\end{definition}
Interleaved Gabidulin codes fulfill the Singleton-like bound in the rank metric with equality, i.e. we have that $\RankDist(\IntGab{\bm{\beta},\intOrder;n,k})=n-k+1$, and thus are \ac{MRD} codes.

Suppose we transmit a codeword
\begin{equation}
 \mat{C}=
 	\begin{pmatrix}
 	 \opev{f^{(1)}}{\bm{\beta}}{1}
 	 \\
 	 \opev{f^{(2)}}{\bm{\beta}}{1} 
 	 \\
 	 \vdots 
 	 \\
 	 \opev{f^{(\intOrder)}}{\bm{\beta}}{1}
 	\end{pmatrix}\in\IntGab{\bm{\beta},\intOrder;n,k}
\end{equation}
and receive
\begin{equation}
 \mat{R}=\mat{C}+\mat{E}
\end{equation}
where $\RankWeight(\mat{E})=t$.

\begin{definition}[Generalized Operator Vector Evaluation Map]\label{def:gen_op_vec_eval}
 Given an interpolation point set $\intPointSet=\{(p_{i,0},p_{i,1},\dots,p_{i,\intParam}):i=0,\dots,n-1\}\subseteq\Fqm^{\intParam+1}$, a vector $\Q\in\SkewPolyringZeroDer^{\intOrder+1}$, and a vector $\vec{a}\in\Fqm^{n}$ containing the generalized operator evaluation parameters, we define the generalized vector evaluation maps as
\begin{equation}\label{eq:def_vec_op_eval}
 \vecOpEv{i}{\vec{Q}}{a_i}\defeq \sum_{j=0}^{\intParam}\opev{Q_j}{p_{i,j}}{a_i}, \quad\forall i=0,\dots,n-1.
\end{equation}
\end{definition}

For interleaved Gabdulin codes, the interpolation point set is
\begin{equation}
	\intPointSet=\{(\beta_i,r_{0,i},\dots,r_{\intOrder-1,i}):i=0,\dots,n-1\}\subset\Fqm^{\intOrder+1}
\end{equation}
with $\vec{a}=(1,1,\dots,1)\in\Fqm^n$ being a vector containing the corresponding generalized operator evaluation parameters.
Equipped with these definitions the interpolation problem ~\cite[Problem~1]{WachterzehZeh-ListUniqueErrorErasureInterpolationInterleavedGabidulin_DCC2014} for decoding interleaved Gabidulin codes can be stated as follows.

\begin{problem}[Vector Interpolation Problem]\label{prob:intProblemIGab}
Given the integer $\intParam\in\mathbb{Z}_+$, a set of $\Fqm$-linear vector evaluation maps $\set{E}^{op}=\{\vecOpEvNoInput{0},\dots,\vecOpEvNoInput{n-1}\}$
as defined in~\eqref{eq:def_vec_op_eval}, the vector $\a=(1,1,\dots,1)\in\Fqm^n$ and a vector $\vec{w}=(0,k-1,\dots,k-1)\in\mathbb{Z}_+^{\intParam+1}$ compute a $\vec{w}$-ordered weak-Popov Basis for the left $\SkewPolyringZeroDer$-module
\begin{equation}\label{eq:int_module_IGab}
    \bar{\module{K}}_{n-1}=\{\vec{b}\in\SkewPolyringZeroDer^{\intParam+1}:\vecOpEv{i}{\vec{b}}{a_i}=0,\,\forall i=0,\dots,n-1\}.
\end{equation}
\end{problem}
A nonzero solution of Problem~\ref{prob:intProblemIGab} exists if the degree constraint satisfies (see~\cite{WachterzehZeh-ListUniqueErrorErasureInterpolationInterleavedGabidulin_DCC2014})
\begin{equation}
	\degConstraint=\left\lceil\frac{n+\intOrder(k-1)+1}{\intOrder+1}\right\rceil.
\end{equation}

\begin{proposition}\label{eq:comp_skew_KNH_gen_op}
 Problem~\ref{prob:intProblemIGab} can be solved by Algorithm~\ref{alg:skewMultVarKNH} in $\oh{\intOrder^2n^2}$ operations in $\Fqm$.
\end{proposition}

\begin{proof}
 Problem~\ref{prob:intProblemIGab} is an instance of the general vector interpolation problem (Problem~\ref{prob:generalIntProblem}) which can be solved by Algorithm~\ref{alg:skewMultVarKNH} (see Lemma~\ref{lem:correctness_skewMultVarKNH}).
 By Proposition~\ref{prop:complSkewMultVarKNH}, Algorithm~\ref{alg:skewMultVarKNH} requires $\oh{sn}$ evaluation maps $\vecOpEvNoInput{i}$ of a vector in $\SkewPolyringZeroDer^{\intParam+1}_{\leq n}$, which requires $\oh{\intParam n}$ operations in $\Fqm$ each.
 Overall, the computation of the evaluation maps requires $\oh{\intParam^2 n^2}$ operations in $\Fqm$.
 The computation of the $n$ multiplications of a monic degree-1 skew polynomial with a vector from $\SkewPolyringZeroDer$ (Line~\ref{alg1:degreeinc}) requires $\oh{\intParam n^2}$ operations in $\Fqm$ in total. 
 The $\oh{\intOrder n}$ multiplications of an element from $\Fqm$ and a vector from $\SkewPolyringZeroDer^{\intParam+1}_{\leq n}$ require at most $\oh{\intParam^2 n^2}$ operations in $\Fqm$.
 Therefore we conclude that Algorithm~\ref{alg:skewMultVarKNH} can solve Problem~\ref{prob:intProblemIGab} requiring at most $\oh{\intOrder^2n^2}$ operations in $\Fqm$. 
 \qed
\end{proof}

If the error weight $t=\RankWeight(\mat{E})$ satisfies
\begin{equation}\label{eq:dec_cond_IGab}
	t<\frac{\intOrder}{\intOrder+1}(n-k+1)
\end{equation}
it can be shown that (see~\cite[Theorem~1]{WachterzehZeh-ListUniqueErrorErasureInterpolationInterleavedGabidulin_DCC2014})
\begin{equation}\label{eq:dec_condition}
	P(x)=Q_0(x)+\sum_{j=1}^{\intOrder}Q_j(x)f^{(j)}(x)=0.
\end{equation}

The root-finding step consists of finding all polynomials $f^{(1)},\dots,f^{(\intOrder)}\in\SkewPolyringZeroDer_{<k}$ that satisfy~\eqref{eq:dec_condition}.
This task can be accomplished by the minimal approximant basis methods in~\cite{bartz2019fast,bartz2021fast} requiring at most $\softoh{\intOrder^\matmult\OMul{n}}$ operations in $\Fqm$. 

\subsubsection{Solving the Interpolation Step via the Fast \ac{KNH} Interpolation}

For an interpolation point set $\intPointSet=\{\vec{p}_0,\vec{p}_1,\dots,\vec{p}_{n-1}\}\subseteq\Fqm^{\intParam+1}$ define the vectors of minimal polynomials with respect to the generalized operator evaluation for $0\leq i,j\leq n-1$ and $i\leq j$ as
\begin{equation}\label{eq:def_gen_op_minpoly_vecs}
    \vecMinpolyOp{[i,j]}{a_i}\defeq\left(\minpolyOp{\{p_{i,0},\dots,p_{j,0}\}}{a_i},\minpolyOp{\{p_{i,1},\dots,p_{j,1}\}}{a_i},\dots,\minpolyOp{\{p_{i,\intParam},\dots,p_{j,\intParam}\}}{a_i}\right).
\end{equation}

\begin{lemma}\label{lem:mod_vec_op_eval}
Let $\set{E}^{op}=(\vecOpEvNoInput{0},\ldots,\vecOpEvNoInput{n-1})$ be an ordered set of skew vector evaluation maps as defined in~\eqref{eq:def_vec_op_eval} and let $\set{E}_{[i,j]}^{op}=(\vecOpEvNoInput{i},\ldots,\vecOpEvNoInput{j})\subseteq\set{E}^{op}$.
Then for any $\vec{Q}\in\SkewPolyringZeroDer^{\intParam+1}$ we have that
\begin{equation}\label{eq:mod_vec_gen_op_eval}
    \vecOpEv{l}{\vec{Q}}{a_l}=\vecOpEv{l}{\vec{Q}\modr\vecMinpolyOp{[i,j]}{\vec{a}}}{a_l},\quad\forall l=i,\dots,j
\end{equation}
where $\vec{a}=(a_i,\dots,a_j)$ contains the corresponding general operator evaluation parameters.
\end{lemma}

\begin{proof}
The lemma follows directly by applying the result from Lemma~\ref{lem:mod_op_eval} to the elementary evaluations in the skew vector operator evaluation maps defined in~\eqref{eq:def_vec_op_eval}.
\qed
\end{proof}

An important consequence of Lemma~\ref{lem:mod_vec_op_eval} is, that the generalized operator vector evaluation maps from Definition~\ref{def:gen_op_vec_eval} and the minimal polynomial vectors in~\eqref{eq:def_gen_op_minpoly_vecs} fulfill Assumption~\ref{ass:mod_vectors}.
Hence, we can solve Problem~\ref{prob:intProblemIGab} by calling Algorithm~\ref{alg:skewIntTree} with $\set{E}^{op}$, basis $\mat{I}_{\intOrder+1}$ and initial degrees $(0, k-1, \dots, k-1)$.

\subsubsection{Complexity Analysis}

We now perform a complexity analysis of Algorithm~\ref{alg:precomputeMinVectors} for the generalized operator evaluation maps defined in~\eqref{eq:def_vec_op_eval} ($\SkewPolyringZeroDer$-case).

\begin{lemma}[Complexity of Computing Minimal Polynomial Vectors]\label{lem:comp_min_vecs_gen_op}
 Algorithm~\ref{alg:precomputeMinVectors} constructs the (ordered) set $\set{M}^{\text{op}}\subset\SkewPolyringZeroDer^{\intOrder+1}$ containing the minimal polynomial vectors defined as
 \begin{equation}
 	\set{M}^{\text{op}}\defeq
    \left(\vecMinpolyOp{[0,n-1]}{\vec{a}},\vecMinpolyOp{[0,\lfloor n/2\rfloor-1]}{\vec{a}},\vecMinpolyOp{[\lfloor n/2\rfloor,n-1]}{\vec{a}},\dots,\vecMinpolyOp{[n-1,n-1]}{\vec{a}}\right)
 \end{equation}
 in $\softoh{\intOrder\OMul{n}}$ operations in $\Fqm$.
\end{lemma}

\begin{proof}
 Algorithm~\ref{alg:precomputeMinVectors} is a generalization of the procedure in~\cite[Theorem~3.2.7]{caruso2017new} to construct a single minimal polynomial, which requires $\softoh{\OMul{n}}$ operations in $\Fqm$.
 Hence, the overall complexity of Algorithm~\ref{alg:precomputeMinVectors} is in the order of $\softoh{\intOrder\OMul{n}}$ operations in $\Fqm$.
 \qed
\end{proof}

\begin{theorem}[Computational Complexity]\label{thm:comp_fast_KNH_IGab}
 Algorithm~\ref{alg:skewIntTree} solves Problem~\ref{prob:intProblemIGab} over $\SkewPolyringZeroDer^{\intOrder+1}$ in $\softoh{\intOrder^\matmult\OMul{n}}$ operations in $\Fqm$.
\end{theorem}

\begin{proof}
  Let $C(n)$ denote the complexity on $n$ input points without the cost of $\textsf{LinInterpolatePoint}$.
  We have that $\deg(\mat{T}_1), \deg(\mat{T}_2)\leq n/2$ and thus the product of $\mat{T}_2\mat{T}_1$ in Line~\ref{line:T2T1} requires $\softoh{\intOrder^\matmult\OMul{n/2}}\in\softoh{\intOrder^\matmult\OMul{n}}$ operations in $\Fqm$.
  By the master theorem we have $C(n)=2C(n/2)+\softoh{\intParam^\matmult n/2}$ implying that $C(n)\in\softoh{\intParam^\matmult n}$.
  The complexity of $\textsf{LinInterpolatePoint}$ is dominated by $\intParam+1$ univariate skew polynomials of degree less than $1$, which requires $\oh{\intParam^2}$ operations in $\Fqm$.  
  The routine $\textsf{LinInterpolatePoint}$ is called $n$ times yielding $\oh{\intParam^2n}$ operations in total.
  By Lemma~\ref{lem:comp_min_vecs_gen_op} all minimal polynomial vectors required in Lines \ref{line:VecMod1} and \ref{line:VecMod2} of Algorithm~\ref{alg:skewIntTree} can be pre-computed in $\softoh{\intOrder\OMul{n}}$ operations in $\Fqm$.
  One (right) modulo operation in $\SkewPolyringZeroDer^{\intOrder+1}$ requires $\oh{\intOrder\OMul{n}}$ operations in $\Fqm$. Therefore, Lines~\ref{line:VecMod1} and~\ref{line:VecMod2} require $\softoh{\intOrder^2\OMul{n}}$ operations in $\Fqm$ each.
 \qed
\end{proof}

\begin{remark}[Practical Consideration]
 Note, that in the $\SkewPolyringZeroDer$-case the $i$-th generalized operator vector evaluation map $\vecOpEv{i}{x\vec{\Q}}{a_i}$ of $x\Q$ can be computed efficiently from $\vecOpEv{i}{\vec{\Q}}{a_i}$ requiring only one application of the automorphism $\aut$ and one multiplication since by the product rule~\cite{martinez2019private} we have that
 \begin{equation}
  \vecOpEv{i}{x\vec{\Q}}{a_i}=\aut\left(\vecOpEv{i}{\vec{\Q}}{a_i}\right)a_i,
  \quad\forall i=0,\dots,n-1.
 \end{equation} 
\end{remark}

\subsection{Interpolation-Based Decoding of Interleaved Linearized Reed--Solomon Codes}

Linearized Reed--Solomon \acused{LRS}(\ac{LRS}) codes are codes that have distance properties with respect to the \emph{sum-rank} metric and were introduced in~\cite{martinez2018skew} and also considered in~\cite{caruso2019residues}.
Recently, codes in the sum-rank metric gained attraction since they generalize several code families in the Hamming metric and the rank metric, such as Reed--Solomon codes and Gabidulin codes, and have potential applications in code-based quantum-resistant cryptosystems~\cite{puchinger2022generic}.
Further applications include the construction of space-time codes~\cite{lu2005unified}, locally repairable codes with maximal recoverability~\cite{martinez2019universal} (also known as partial MDS codes) and error control for multishot network coding~\cite{martinez2019reliable}.
Interpolation-based decoding of $\intOrder$-\emph{interleaved} \ac{LRS} codes, which allows for correcting errors beyond the unique decoding radius (up to $\frac{\intOrder}{\intOrder+1}(n-k+1)$) in the sum-rank metric, was recently considered in~\cite{bartz2022fast}.
In the following we show, how Algorithm~\ref{alg:skewIntTree} can be used to solve the interpolation step in the interpolation-based decoder for \ac{ILRS} from~\cite{bartz2022fast} efficiently.

The \emph{sum-rank weight} of a vector $\vec{x}=\left(\vec{x}^{(1)}\mid\vec{x}^{(2)}\mid\dots\mid\allowbreak\vec{x}^{(\shots)}\right)\in\Fqm^n$, where $\vec{x}^{(i)}\in\Fqm^{n_l}$ for all $l=1,\dots,\shots$, is defined as (see~\cite{lu2005unified,nobrega2010multishot})
 \begin{equation}
 	\SumRankWeight(\vec{x})\defeq\sum_{l=1}^{\shots}\rk_q\left(\vec{x}^{(l)}\right).
 \end{equation}
For any $\vec{x}\in\Fqm^n$ we have that $\SumRankWeight(\vec{x})\leq\HammingWeight(\vec{x})$ (see~\cite{martinez2016similarities,martinez2018skew}).
The \emph{sum-rank distance} between two vectors $\vec{x},\vec{y}\in\Fqm^n$ is then
\begin{equation}
	\SumRankDist(\vec{x},\vec{y})\defeq\SumRankWeight(\vec{x}-\vec{y})=\sum_{l=1}^{\shots}\rk_q\left(\vec{x}^{(l)}-\vec{y}^{(l)}\right).
\end{equation}

The sum-rank weight of a matrix $\mat{X}=(\mat{X}^{(1)} \mid \mat{X}^{(2)} \mid \dots \mid \mat{X}^{(\shots)})\in\Fqm^{\intOrder\times n}$ is defined as 
\begin{equation}
  \SumRankWeight(\mat{X})\defeq\sum_{i=1}^{\shots}\rk_q\left(\mat{X}^{(i)}\right),
\end{equation}
where $\mat{X}^{(i)}\in\Fqm^{\intOrder\times n_i}$ for all $i=1,\dots,n_i$.
The sum-rank distance between two matrices $\mat{X},\mat{Y}\in\Fqm^{\intOrder\times n}$ is then defined as
\begin{equation}
  \SumRankDist(\mat{X},\mat{Y})\defeq\SumRankWeight(\mat{X}-\mat{Y})=\sum_{i=1}^{\shots}\rk_q\left(\mat{X}^{(i)}-\mat{Y}^{(i)}\right).
\end{equation}

As channel model we consider the \emph{sum-rank channel} 
\begin{equation}
	\vec{r}=\vec{c}+\vec{e}
\end{equation}
where the error vector
\begin{equation}
 \vec{e}=(\vec{e}^{(1)}\mid\vec{e}^{(2)}\mid\dots\mid\vec{e}^{(\shots)})
\end{equation}
has sum-rank weight $\SumRankWeight(\vec{e})=t$.

\begin{definition}[Interleaved Linearized Reed--Solomon Code]
 Let $\aut$ be the Frobenius automorphism of $\Fqm$ defined as $\aut(\cdot)=\cdot^q$.
 Let $\bm{\xi}=(\xi_0,\xi_1,\dots,\xi_{\shots-1})\in\Fqm^\shots$ be a vector containing representatives from different conjugacy classes of $\Fqm$.
 Let the vectors $\bm{\beta}^{(l)}=\Big(\beta_0^{(l)},\allowbreak\dots,\allowbreak\beta_{n_l-1}^{(l)}\Big)\in\Fqm^{n_l}$ contain $\Fq$-linearly independent elements from $\Fqm$ for all $l=1,\dots,\shots$ and define the vector $\bm{\beta}=\left(\bm{\beta}^{(1)},\bm{\beta}^{(2)},\dots,\bm{\beta}^{(\shots)}\right)\in\Fqm^n$.
 An $\intOrder$-interleaved linearized Reed--Solomon (ILRS) code of length $n=n_1+n_2+\dots+n_\shots$ and dimension $k$ is defined as
 \begin{equation}
 	\intLinRS{\bm{\beta},\shots,\intOrder;n,k}=\left\{\left(
 	\begin{array}{c|c|c}
 	 \opev{f^{(1)}}{\bm{\beta}^{(1)}}{\xi_0} \, & \quad\dots\quad & \, \opev{f^{(1)}}{\bm{\beta}^{(\shots)}}{\xi_{\shots-1}}
 	 \\
 	 \opev{f^{(2)}}{\bm{\beta}^{(1)}}{\xi_0} & \dots & \opev{f^{(2)}}{\bm{\beta}^{(\shots)}}{\xi_{\shots-1}}
 	 \\
 	 \vdots & \ddots & \vdots
 	 \\
 	 \opev{f^{(\intOrder)}}{\bm{\beta}^{(1)}}{\xi_0} & \dots & \opev{f^{(\intOrder)}}{\bm{\beta}^{(\shots)}}{\xi_{\shots-1}}
 	\end{array}
 	\right): \begin{array}{c}f^{(j)}\in\SkewPolyringZeroDer_{<k}, \\\forall j\in\intervallincl{1}{\intOrder}\end{array}\right\}.
 \end{equation}
\end{definition}

The code rate and the minimum sum-rank distance of~\ac{ILRS} codes is $R=\frac{k}{n}$ and $d=n-k+1$ (see~\cite{bartz2022fast}), respectively. \ac{ILRS} codes fulfill the Singleton-like bound in the sum-rank metric with equality and thus are~\ac{MSRD} codes~\cite{martinez2018skew}. 

Suppose we transmit a codeword
\begin{equation}
	\mat{C}=
	\left(
	\begin{array}{c|c|c}
	  \opev{f^{(1)}}{\bm{\beta}^{(1)}}{\xi_0} \, & \quad\dots\quad & \, \opev{f^{(1)}}{\bm{\beta}^{(\shots)}}{\xi_{\shots-1}}
 	 \\
 	 \opev{f^{(2)}}{\bm{\beta}^{(1)}}{\xi_0} & \dots & \opev{f^{(2)}}{\bm{\beta}^{(\shots)}}{\xi_{\shots-1}}
 	 \\
 	 \vdots & \ddots & \vdots
 	 \\
 	 \opev{f^{(\intOrder)}}{\bm{\beta}^{(1)}}{\xi_0} & \dots & \opev{f^{(\intOrder)}}{\bm{\beta}^{(\shots)}}{\xi_{\shots-1}}
	\end{array}
	\right)\in\intLinRS{\bm{\beta},\shots,\intOrder;n,k}
\end{equation}
and receive a matrix 
\begin{equation}
 \mat{R}=\left(\mat{R}^{(1)}\mid\mat{R}^{(2)}\mid\dots\mid\mat{R}^{(\shots)}\right)=\mat{C}+\mat{E}
\end{equation}
with $\R^{(l)}\in\Fqm^{\intOrder\times n_l}$ for all $l=1,\dots,\shots$ where the error matrix $\E=\left(\E^{(1)}\mid\E^{(2)}\mid\dots\mid\E^{(\shots)}\right)$ has sum-rank weight $\SumRankWeight(\mat{E})=t$.
In order to simplify the notation we index the entries in $\mat{R}$ as $r_{j,i}$ for $j=0,\dots,\intOrder-1$ and $i=0,\dots,n-1$.

Further, we define the vector 
\begin{equation}\label{eq:def_eval_param_vec}
    \vec{a}=(\vec{a}^{(1)},\vec{a}^{(2)},\dots,\vec{a}^{(\shots)})\in\Fqm^n
\end{equation}
where $\vec{a}^{(l)}=(\xi_{l-1},\dots,\xi_{l-1})\in\Fqm^{n_l}$ for all $l=1,\dots,\shots$. 

For interleaved linearized Reed--Solomon codes, the interpolation point set is
\begin{equation}
    \intPointSet=\{(\beta_i,r_{0,i},\dots,r_{\intOrder-1,i}):i=0,\dots,n-1\}\subset\Fqm^{\intOrder+1}.
\end{equation}

\begin{problem}[Vector Interpolation Problem]\label{prob:intProblemILRS}
Given the integer $\intParam\in\mathbb{Z}_+$, a set of $\Fqm$-linear vector evaluation maps $\set{E}^{op}=\{\vecOpEvNoInput{0},\dots,\vecOpEvNoInput{n-1}\}$
as defined in~\eqref{eq:def_vec_op_eval}, a vector $\a=(a_0,a_1,\dots,\allowbreak a_{n-1})\in\Fqm^n$ as defined in~\eqref{eq:def_eval_param_vec} and a vector $\vec{w}=(0,k-1,\dots,k-1)\in\mathbb{Z}_+^{\intParam+1}$ compute a $\vec{w}$-ordered weak-Popov Basis for the left $\SkewPolyringZeroDer$-module
\begin{equation}\label{eq:int_module_ILRS}
    \bar{\module{K}}_{n-1}=\{\vec{b}\in\SkewPolyringZeroDer^{\intParam+1}:\vecOpEv{i}{\vec{b}}{a_i}=0,\,\forall i=0,\dots,n-1\}.
\end{equation}
\end{problem}

A nonzero solution of Problem~\ref{prob:intProblemILRS} exists if the degree constraint satisfies
\begin{equation}
	\degConstraint=\left\lceil\frac{n+\intOrder(k-1)+1}{\intOrder+1}\right\rceil.
\end{equation}
By following the ideas in Proposition~\ref{eq:comp_skew_KNH_gen_op} for general evaluation parameters we see that Problem~\ref{prob:intProblemILRS} can be solved by Algorithm~\ref{alg:skewMultVarKNH} in $\oh{\intOrder^2n^2}$ operations in $\Fqm$.

If the sum-rank weight of the error $t=\SumRankWeight(\mat{E})$ satisfies
\begin{equation}\label{eq:dec_cond_ILRS}
 t<\frac{\intOrder}{\intOrder+1}(n-k+1)
\end{equation}
we have that (see~\cite[Theorem~2]{bartz2022fast})
\begin{equation}\label{eq:dec_condition_ILRS}
	P(x)=Q_0(x)+\sum_{j=1}^{\intOrder}Q_j(x)f^{(j)}(x)=0.
\end{equation}

The root-finding problem consists of finding all polynomials $f^{(1)},\dots,f^{(\intOrder)}\in\SkewPolyringZeroDer_{<k}$ that satisfy~\eqref{eq:dec_condition_ILRS}.
The root-finding problem can be solved efficiently by the minimal approximant basis methods in~\cite{bartz2019fast,bartz2021fast} requiring at most $\softoh{\intOrder^\matmult\OMul{n}}$ operations in $\Fqm$. 

\subsubsection{Solving the Interpolation Step via the Fast \ac{KNH} Interpolation}

Since Lemma~\ref{lem:mod_vec_op_eval} holds for arbitrary evaluation parameters $\vec{a}$ we have that the minimal polynomial vectors in~\eqref{eq:def_gen_op_minpoly_vecs} fulfill Assumption~\ref{ass:mod_vectors}.
Consequently, we can solve Problem~\ref{prob:intProblemILRS} by calling Algorithm~\ref{alg:skewIntTree} with $\set{E}^{op}$ as defined in~\eqref{eq:def_vec_op_eval}, the basis $\mat{I}_{\intOrder+1}$ and initial degrees $(0,k-1,\dots,k-1)$.

\subsubsection{Complexity Analysis}

Note, that Problem~\ref{prob:intProblemIGab} is an instance of Problem~\ref{prob:intProblemILRS} with particular evaluation parameters.
Therefore, the complexity follows directly from Theorem~\ref{thm:comp_fast_KNH_IGab}.

\begin{corollary}[Computational Complexity]\label{col:comp_fast_KNH_ILRS}
 Algorithm~\ref{alg:skewIntTree} solves Problem~\ref{prob:intProblemILRS} over $\SkewPolyringZeroDer^{\intOrder+1}$ in $\softoh{\intOrder^\matmult\OMul{n}}$ operations in $\Fqm$.
\end{corollary}

\subsection{Interpolation-Based Decoding of Interleaved Skew Reed--Solomon Codes}

In this section we consider decoding of Skew Reed--Solomon\acused{SRS} (\ac{SRS}) codes with respect to the \emph{skew metric}, which was introduced in~\cite{martinez2018skew}.
Decoding schemes for \ac{SRS} codes that allow for correcting error of skew weight up to $\lfloor\frac{n-k}{2}\rfloor$ were presented in~\cite{martinez2018skew,boucher2018algorithm,liu2015construction,bartz2021fast}.
Interpolation-based decoding of \emph{interleaved} \ac{SRS} \ac{ISRS} codes that allows for decoding errors of skew weight up to $\frac{\intOrder}{\intOrder+1}(n-k+1)$, where $\intOrder$ is the \emph{interleaving order}, was considered in~\cite{bartz2022fast}.
The interpolation scheme can be either used as a list decoder or as a probabilistic-unique decoder.

In this work we consider the definition of the skew weight from~\cite[Proposition~1]{boucher2018algorithm}.
Let the vector $\vec{b}\in\Fqm^n$ contain $P$-independent elements and let $\set{B}$ be the $P$-closed set generated by the elements in $\b$.
Then, the \emph{skew weight} of a vector $\vec{x}\in\Fqm^n$ with respect to $\set{B}$ is defined as
\begin{equation}\label{eq:def_skew_weight}
 \SkewWeight(\vec{x})\defeq \deg\left(\lclm{x-\conjg{b_i}{x_i}}{\mystack{1\leq i\leq n-1}{x_i\neq0}}\right).
\end{equation}
For simplicity we omit the dependence on $\set{B}$ in the definition of $\SkewWeight(\cdot)$ as it will be clear from the context.
Similar to the rank and the sum-rank weight we have that $\SkewWeight(\vec{x})\leq\HammingWeight(\vec{x})$ for all $\vec{x}\in\Fqm^n$ (see~\cite{martinez2018skew}).
The \emph{skew distance} between two vectors $\vec{x},\vec{y}\in\Fqm^n$ is defined as
\begin{equation}\label{eq:def_skew_distance}
 \SkewDist(\vec{x},\vec{y})\defeq\SkewWeight(\vec{x}-\vec{y}).
\end{equation}
As a channel model we consider the \emph{skew metric channel}
\begin{equation}
 \vec{r}=\vec{c}+\vec{e}
\end{equation}
where the error vector $\vec{e}$ has skew weight $\SkewWeight(\vec{e})=t$.

We define (vertically) interleaved skew Reed--Solomon (ISRS) codes as follows.
\begin{definition}[Interleaved Skew Reed--Solomon Code]
 Let $\vec{b}=(b_0,b_1,\dots,b_{n-1})\in\Fqm^n$ contain $P$-independent elements from $\Fqm$. 
 For a fixed integer $k\leq n$, a vertically $\intOrder$-interleaved skew Reed--Solomon (ISRS) code of length $n$ and dimension $k$ is defined as
 \begin{equation}\label{eq:defSkewRS}
 	\intSkewRS{\vec{b},\intOrder;n,k}
 	=\left\{
 	\begin{pmatrix}
 	 \remev{f^{(1)}}{\vec{b}}
 	 \\
 	 \remev{f^{(2)}}{\vec{b}}
 	 \\
 	 \vdots 
 	 \\
 	 \remev{f^{(\intOrder)}}{\vec{b}}
 	\end{pmatrix}:f^{(j)}\in\SkewPolyringZeroDer_{<k},\forall j\in\intervallincl{1}{\intOrder}
 	\right\}.
 \end{equation}
\end{definition}

\ac{ISRS} codes are \ac{MSD} codes, i.e. we have that
\begin{equation}
	\SkewDist(\intSkewRS{\vec{b},\intOrder;n,k})=n-k+1.
\end{equation}

Under a fixed basis of $\Fqms$ over $\Fqm$ there is a bijection between a matrix $\mat{C}\in\Fqm^{\intOrder\times n}$ and a vector $\vec{c}\in\Fqms^n$.
Hence, we can represent each codeword $\mat{C}\in\intSkewRS{\vec{b},\intOrder;n,k}$ as a vector $\vec{c}=f(\vec{b})\in\Fqms^n$ where $f\in\Fqms[\x;\aut]$ is the the polynomial obtained by considering the coefficients of $f^{(1)},\dots,f^{(\intOrder)}$ over $\Fqms$.
This relation between $\intOrder$-interleaved evaluation codes over $\Fqm$ and punctured evaluation codes over the bigger field $\Fqms$ is well-known from Reed--Solomon and Gabidulin codes (see e.g.~\cite{sidorenko2008decoding,bartz2017algebraic}).

Suppose we transmit a vector $\vec{c}=f(\vec{b})\in\intSkewRS{\vec{b},\intOrder;n,k}$ and receive a vector
\begin{equation}
 \mat{r}=\mat{c}+\mat{e}\in\Fqms^n
\end{equation}
where $\SkewWeight(\vec{e})=t$.
Let $\mat{R}\in\Fqm^{\intOrder\times n}$ denote the expanded matrix of $\vec{r}\in\Fqms^n$ over $\Fqm$.

\begin{definition}[Remainder Vector Evaluation Map]\label{def:rem_vec_eval}
 Given an interpolation point set $\intPointSet=\{(p_{i,0},p_{i,1},\dots,p_{i,\intParam}):i=0,\dots,n-1\}\subseteq\Fqm^{\intParam+1}$ and a vector $\Q\in\SkewPolyringZeroDer^{\intOrder+1}$ we define the remainder vector evaluation maps as
\begin{equation}\label{eq:def_vec_rem_eval}
 \vecRemEv{i}{\vec{Q}}\defeq \remev{Q_0}{p_{i,0}} + \sum_{j=1}^{\intParam}\remev{Q_j}{\conjg{p_{i,0}}{p_{i,j}}}p_{i,j}, \quad\forall i=0,\dots,n-1.
\end{equation}
\end{definition}

\begin{remark}\label{rem:conjugates_with_zero}
 Note, that $\conjg{p_{i,0}}{p_{i,j}}$ is not defined whenever $p_{i,j}=0$. Similar to~\cite{martinez2019reliable} we define $\remev{Q_j}{\conjg{p_{i,0}}{p_{i,j}}}p_{i,j}=0$ for all $p_{i,j}=0$.
 This definition is very natural in view of the correspondence to the generalized operator evaluation (see Lemma~\ref{lem:rel_remev_opev}) and the product rule (see~\cite[Theorem~2.7]{lam1988vandermonde}).
\end{remark}

For interleaved skew Reed--Solomon codes, the interpolation point set is
\begin{equation}
    \intPointSet=\{(\beta_i,r_{0,i},\dots,r_{\intOrder-1,i}):i=0,\dots,n-1\}\subset\Fqm^{\intOrder+1}.
\end{equation}

\begin{problem}[Vector Interpolation Problem]\label{prob:intProblemISRS}
Given the integer $\intParam\in\mathbb{Z}_+$, a set of $\Fqm$-linear vector evaluation maps $\set{E}^{rem}=\{\vecRemEvNoInput{0},\dots,\vecRemEvNoInput{n-1}\}$
as defined in~\eqref{eq:def_vec_rem_eval} and a vector $\vec{w}=(0,k-1,\dots,k-1)\in\mathbb{Z}_+^{\intParam+1}$ compute a $\vec{w}$-ordered weak-Popov Basis for the left $\SkewPolyringZeroDer$-module
\begin{equation}\label{eq:int_module_ISRS}
    \bar{\module{K}}_{n-1}=\{\vec{b}\in\SkewPolyringZeroDer^{\intParam+1}:\vecRemEv{i}{\vec{b}}=0,\,\forall i=0,\dots,n-1\}.
\end{equation}
\end{problem}

A nonzero solution of Problem~\ref{prob:intProblemISRS} exists if the degree constraint satisfies (see~\cite{bartz2022fast})
\begin{equation}
	\degConstraint=\left\lceil\frac{n+\intOrder(k-1)+1}{\intOrder+1}\right\rceil.
\end{equation}

\begin{proposition}
Problem~\ref{prob:intProblemISRS} can be solved by Algorithm~\ref{alg:skewMultVarKNH} in $\oh{\intOrder^2n^2}$ operations in $\Fqm$.
\end{proposition}
\begin{proof}
 Problem~\ref{prob:intProblemISRS} is an instance of the general vector interpolation problem (Problem~\ref{prob:generalIntProblem}) which can be solved by Algorithm~\ref{alg:skewMultVarKNH} (see Lemma~\ref{lem:correctness_skewMultVarKNH}).
 By Proposition~\ref{prop:complSkewMultVarKNH}, Algorithm~\ref{alg:skewMultVarKNH} requires the computation of $\oh{sn}$ evaluation maps $\vecRemEvNoInput{i}$ of a vector in $\SkewPolyringZeroDer^{\intParam+1}_{\leq n}$, which requires $\oh{\intParam n}$ operations in $\Fqm$ each.
 Overall, the computation of the evaluation maps requires $\oh{\intParam^2 n^2}$ operations in $\Fqm$.
 The computation of the $n$ multiplications of a monic degree-1 skew polynomial with a vector from $\SkewPolyringZeroDer$ (Line~\ref{alg1:degreeinc}) requires $\oh{\intParam n^2}$ operations in $\Fqm$ in total. 
 The $\oh{\intOrder n}$ multiplications of an element from $\Fqm$ and a vector from $\SkewPolyringZeroDer^{\intParam+1}_{\leq n}$ require at most $\oh{\intParam^2 n^2}$ operations in $\Fqm$.
 Therefore we conclude that Algorithm~\ref{alg:skewMultVarKNH} can solve Problem~\ref{prob:intProblemIGab} requiring at most $\oh{\intOrder^2n^2}$ operations in $\Fqm$. 
 \qed
\end{proof}

If the skew weight of the error $t=\SkewWeight(\vec{e})$ satisfies
\begin{equation}\label{eq:dec_cond_SRS}
	t<\frac{\intOrder}{\intOrder+1}(n-k+1)
\end{equation}
it can be shown that (see~\cite{bartz2022fast})
\begin{equation}\label{eq:dec_condition_SRS}
	P(x)=Q_0(x)+\sum_{j=1}^{\intOrder}Q_j(x)f^{(j)}(x)=0.
\end{equation}

The root-finding step consists of finding all polynomials $f^{(1)},\dots,f^{(\intOrder)}\in\SkewPolyringZeroDer_{<k}$ that satisfy~\eqref{eq:dec_condition_SRS}.
The root-finding problem can be solved efficiently by the minimal approximant basis methods in~\cite{bartz2019fast,bartz2021fast} requiring at most $\softoh{\intOrder^\matmult\OMul{n}}$ operations in $\Fqm$. 

\subsubsection{Solving the Interpolation Step via the Fast \ac{KNH} Interpolation}

Define the sets
\begin{equation}\label{def:rem_ev_min_poly_sets}
    \set{P}_{[i,j]}^{(r)}\defeq\left\{\conjg{p_{l,0}}{p_{l,r}}:p_{l,r}\neq0, \forall l=i,\dots,j\right\}
\end{equation}
for all $0\leq i,j\leq n-1$ and $r=1,\dots,s$.
Then the vectors of minimal polynomials with respect to the generalized operator evaluation are defined as
\begin{equation}\label{eq:def_rem_minpoly_vecs}
 \vecMinpolyRem{[i,j]}\defeq
 \left(\minpolyRem{\{p_{i,0},\dots,p_{j,0}\}},\minpolyRem{\set{P}_{[i,j]}^{(1)}},\dots,\minpolyRem{\set{P}_{[i,j]}^{(\intParam)}}\right).
\end{equation}

\begin{lemma}\label{lem:mod_vec_rem_eval}
Let $\set{E}^{rem}=(\vecRemEvNoInput{0},\dots,\vecRemEvNoInput{n-1})$ be an ordered set of skew vector evaluation maps as defined in~\eqref{eq:def_vec_rem_eval} and let $\set{E}_{[i,j]}^{rem}=(\vecRemEvNoInput{i},\ldots,\vecRemEvNoInput{j})\subseteq\set{E}^{rem}$.
Then for any $\vec{Q}\in\SkewPolyringZeroDer^{\intParam+1}$ we have that
\begin{equation}\label{eq:mod_vec_rem_eval}
    \vecRemEv{l}{\vec{Q}}=\vecRemEv{l}{\vec{Q}\modr\vecMinpolyRem{[i,j]}},
    \quad\forall l=i,\dots,j.
\end{equation}
\end{lemma}

\begin{proof}
For the case where the interpolation point $\vec{p}_i$ (corresponding to $\vecRemEvNoInput{i}$) contains only nonzero elements, the lemma follows directly by applying the result from Lemma~\ref{lem:mod_rem_eval} to the elementary evaluations in the skew vector remainder evaluation maps defined in~\eqref{eq:def_vec_rem_eval}.
Since by definition of the sets in~\eqref{def:rem_ev_min_poly_sets} all conjugates $\conjg{p_{i,0}}{p_{i,r}}$ where $p_{i,r}=0$ are excluded, and by definition the evaluation of each $\remev{Q_r}{\conjg{p_{i,0}}{p_{i,r}}}p_{i,r}=0$ for all $p_{i,r}=0$ (see Remark~\ref{rem:conjugates_with_zero}), we have that the statement also holds in this case.
\qed
\end{proof}

An important consequence of Lemma~\ref{lem:mod_vec_rem_eval} is, that the vector remainder evaluation maps from Definition~\ref{def:rem_vec_eval} and the minimal polynomial vectors in~\eqref{eq:def_rem_minpoly_vecs} fulfill Assumption~\ref{ass:mod_vectors}.
Hence, we can solve Problem~\ref{prob:intProblemISRS} by calling Algorithm~\ref{alg:skewIntTree} with $\textsf{SkewInterpolateTree}(\set{E}^{rem}, \mat{I}_{\intOrder+1}, (0,k-1,\dots,k-1))$.

\subsubsection{Complexity Analysis}

\begin{lemma}[Complexity of Computing Minimal Polynomial Vectors]\label{lem:comp_min_vecs_rem}
 Algorithm~\ref{alg:precomputeMinVectors} constructs the (ordered) set containing the minimal polynomial vectors
 \begin{equation}
 	\set{M}^{\text{rem}}\defeq
    \left(\vecMinpolyRem{[0,n-1]},\vecMinpolyRem{[0,\lfloor n/2\rfloor-1]},\vecMinpolyRem{[\lfloor n/2\rfloor,n-1]}\dots,\vecMinpolyRem{[n-1,n-1]}\right)\subset\SkewPolyringZeroDer^{\intOrder+1}
 \end{equation}
 as defined in~\eqref{eq:def_rem_minpoly_vecs} in $\softoh{\intOrder\OMul{n}}$ operations in $\Fqm$.
\end{lemma}

\begin{proof}
 Algorithm~\ref{alg:precomputeMinVectors} is a generalization of the procedure in~\cite[Theorem~3.2.7]{caruso2017new} to construct a single minimal polynomial, which requires $\softoh{\OMul{n}}$ operations in $\Fqm$.
 Hence, the overall complexity of Algorithm~\ref{alg:precomputeMinVectors} is in the order of $\softoh{\intOrder\OMul{n}}$ operations in $\Fqm$.
 \qed
\end{proof}

\begin{theorem}[Computational Complexity]\label{thm:comp_fast_KNH_ISRS}
 Algorithm~\ref{alg:skewIntTree} solves Problem~\ref{prob:intProblemISRS} over $\SkewPolyringZeroDer^{\intOrder+1}$ in $\softoh{\intOrder^\matmult\OMul{n}}$ operations in $\Fqm$.
\end{theorem}

\begin{proof}
 Follows the ideas of the proof of Theorem~\ref{thm:comp_fast_KNH_IGab}.
 \qed
\end{proof}

\begin{remark}[Practical Consideration]
 Note, that in the $\SkewPolyringZeroDer$-case the $i$-th remainder vector evaluation map $\vecRemEv{i}{x\vec{\Q}}$ of $x\Q$ can be computed efficiently from $\vecRemEv{i}{\vec{\Q}}$ requiring only one application of the automorphism $\aut$ and one multiplication since by the product rule~\cite{lam1988vandermonde} we have that
 \begin{equation}
  \vecRemEv{i}{x\vec{\Q}}=\aut\left(\vecRemEv{i}{\vec{\Q}}\right)p_{i,0},
  \quad\forall i=0,\dots,n-1.
 \end{equation} 
\end{remark}

\subsection{Applications to other Related Coding Problems}

Algorithm~\ref{alg:skewIntTree} with generalized operator vector evaluation maps of the form as defined in Definition~\ref{def:gen_op_vec_eval} can be used to solve the interpolation step in other related interpolation-based decoding problems efficiently. The corresponding interpolation problems are instances of Problem~\ref{prob:generalIntProblem} where the minimal polynomial vectors in~\eqref{eq:def_gen_op_minpoly_vecs} fulfill Assumption~\ref{ass:mod_vectors}.

In particular, the interpolation step in decoding of lifted $\intOrder$-interleaved Gabidulin codes in the subspace metric~\cite{bartz2018efficient,bartz2017algebraic} can be performed requiring at most $\softoh{\intOrder^\matmult\OMul{\nReceive}}$ operations in $\Fqm$, where $\nReceive$ denotes the dimension of the received subspace.
Further, the interpolation step in decoding $\foldPar$-folded Gabidulin codes~\cite{mahdavifar2012list,bartz2017algebraic} and \emph{lifted} $\foldPar$-folded Gabidulin codes~\cite{bartz2015list} can be performed in at most $\softoh{\intParam^\matmult\OMul{n}}$ and $\softoh{\intParam^\matmult\OMul{\nReceive}}$ operations in $\Fqm$, respectively, where $\intParam \leq \foldPar$ denotes an interpolation parameter.

Algorithm~\ref{alg:skewIntTree} can also be used to solve the interpolation step for decoding lifted $\intOrder$-interleaved \ac{LRS} codes in the sum-subspace metric~\cite{bartz2021decoding} requiring at most $\softoh{\intOrder^\matmult\OMul{\nReceive}}$ operations in $\Fqm$, where $\nReceive$ denotes the sum of the dimensions of the received subspaces over all shots.

Algorithm~\ref{alg:skewIntTree} can also solve the interpolation step for interpolation-based decoding of $\foldPar$-folded l\ac{LRS} codes~\cite{hormann2021efficient} requiring at most $\softoh{\intParam^\matmult\OMul{n}}$ operations in $\Fqm$, where $\intParam \leq \foldPar$ is a decoding parameter and $n$ is the length of the unfolded code.

Since, unlike the minimal approximant basis approach in~\cite{bartz2021fast,bartz2022fast}, Algorithm~\ref{alg:skewIntTree} has no requirements on the interpolation points (related to the evaluation maps), it can be applied in a straight-forward manner to decoding problems of lifted Gabidulin and \ac{LRS} code variants in the (sum-)subspace metric.

\section{Conclusion}\label{sec:conclusion}

 We proposed a fast divide-and conquer variant of the \acf{KNH} interpolation over free modules over skew polynomial rings.
 We showed how the proposed \ac{KNH} interpolation can be used to solve the interpolation step of interpolation-based decoding of interleaved Gabidulin, linearized Reed--Solomon and skew Reed--Solomon codes and variants thereof efficiently requiring at most $\softoh{\intOrder^\matmult\OMul{n}}$ operations in $\Fqm$, where $n$ is the length of the code, $\intOrder$ the interleaving order, $\OMul{n}$ the complexity for multiplying two skew polynomials of degree at most $n$, $\matmult$ the matrix multiplication exponent and $\softoh{\cdot}$ the soft-$O$ notation which neglects log factors.
 The computational complexity of the proposed fast \ac{KNH} variant coincides with the complexity of the currently fastest interpolation algorithms for skew polynomial rings, where the proposed variant relies on the well-known (bottom-up) \ac{KNH} interpolation algorithm instead of quite involved (top-down) minimal approximant bases techniques. 
 The proposed results also hold for codes defined over general skew polynomials rings (with derivations), except for the complexity analysis.
 Due to the bottom-up nature of the proposed KNH interpolation there are no requirements on the interpolation points and thus no pre-processing of the interpolation points, which may be required for the top-down minimal approximant bases approaches, is necessary. 

 

\vfill

\bibliographystyle{spmpsci}      



\end{document}